\documentclass[draft]{amsart}
\usepackage{mathrsfs}
\usepackage{amssymb}
\usepackage{amsmath}
\usepackage{amsfonts}
\usepackage{amsfonts,enumerate}
\usepackage{pifont}
\usepackage{enumerate}
\usepackage{graphics}
\usepackage{verbatim}
\usepackage{amsthm}

\numberwithin{equation}{section}
\newtheorem{theorem}{Theorem}[section]

\newtheorem{definition}{Definition}[section]
\newtheorem{lemma}{Lemma}[section]

\newtheorem{remark}{Remark}[section]

\newcommand{\8}{\infty}
\newcommand{\el}{\ell}

\newcommand{\be}{\begin{eqnarray*}}
\newcommand{\ee}{\end{eqnarray*}}
\newcommand{\beq}{\begin{equation}}
\newcommand{\eeq}{\end{equation}}
\newcommand{\beqn}{\begin{equation*}}
\newcommand{\eeqn}{\end{equation*}}
\newcommand{\bsp}{\begin{split}}
\newcommand{\esp}{\end{split}}

\begin{document}
\title{Local well-posedness for Gross-Pitaevskii hierarchies}


\author{Zeqian Chen}

\address{Wuhan Institute of Physics and Mathematics, Chinese Academy of Sciences, West District 30, Xiao-Hong-Shan, Wuhan 430071, China}





\date{}
\maketitle
\markboth{Z. Chen}%
{Gross-Pitaevskii hierarchies}

\begin{abstract}
We consider the Cauchy problem for the Gross-Pitaevskii infinite linear
hierarchy of equations on $\mathbb{R}^n.$ By introducing a (F)-norm in certain Sobolev type spaces of sequences of marginal density matrices, we establish local existence, uniqueness and stability of solutions. Explicit space-time type estimates for the solutions are obtained as well. In particular, this (F)-norm is compatible with the usual Sobolev space norm whenever the initial data is factorized.
\end{abstract}


\section{Introduction}\label{Intro}

The time-dependent Gross-Pitaevskii (GP) equation describes the dynamics of initially trapped Bose-Einstein condensates (cf. \cite{DGPS}). Precisely, in units where $\hbar =1$ and the mass of the bosons $m=1/2,$ the condensate wave function at time $t,$ $\varphi_t = \varphi_t (\mathbf{r}),$ $\mathbf{r} \in \mathbb{R}^3,$ satisfies the GP equation
\beq\label{eq:GPeq}
\mathrm{i} \partial_t \varphi_t = - \Delta_\mathbf{r} \varphi_t + \sigma |\varphi_t |^2 \varphi_t
\eeq
with the normalization $\int | \varphi_t (\mathbf{r}) |^2 d ^3 \mathbf{r} = 1,$ where $\sigma = 8 \pi N a$ is the coupling constant with $N$ being the number of particles and $a$ the scattering length of the interaction potential. This fact was rigorously proved by Erd\"{o}s, Schlein, and Yau in \cite{ESY2007a, ESY2007b, ESY2010}, starting from a many-body bosonic Schr\"{o}dinger equation with a short-scale repulsive interaction in the dilute limit.

One key ingredient in the argument of \cite{ESY2007a, ESY2007b, ESY2010} is to prove the uniqueness of solutions to the following infinite hierarchy (commonly called the Gross-Pitaevskii hierarchy):
\beq\label{eq:GPhierarchy}
\mathrm{i} \partial_t \gamma^{(k)}_{t} = \sum^k_{j =1} [-\triangle_{\mathbf{r}_j}, \gamma^{(k)}_{t} ] + B^{(k)} \gamma^{(k+1)}_{t},\quad k=1,2,\ldots,
\eeq
with the collision operator $B^{(k)} = \sum^k_{j =1} B_{j, k}$ being defined according to
\beq\label{eq:CollisionOper}
B_{j,k} \gamma^{(k+1)} = \mathrm{T r}_{k+1} [\delta (\mathbf{r}_j - \mathbf{r}_{k+1}), \gamma^{(k+1)} ],\quad j=1,\ldots, k,
\eeq
where $\gamma^{(k)}$ denotes $k$-partite density operator in $\mathbb{R}^3,$ i.e., a nonnegative operator on $L^2 (\mathbb{R}^3)^{\hat{\otimes}^k}$ with $\mathrm{Tr} \gamma^{(k)} =1.$ The uniqueness result of Erd\"{o}s, Schlein, and Yau states that: Given a family of densities $\Gamma = (\gamma^{(k)})_{k \ge 1}$ such that
\beq\label{eq:ESYnorm}
||| \gamma^{(k)} |||_k : =  \mathrm{T r} \big [ |(1- \Delta_{\mathbf{r}_1})^{\frac{1}{2}}\cdots (1- \Delta_{\mathbf{r}_k})^{\frac{1}{2}} \gamma^{(k)} (1- \Delta_{\mathbf{r}_1})^{\frac{1}{2}}\cdots (1- \Delta_{\mathbf{r}_k})^{\frac{1}{2}} | \big ] \le C^k,\quad \forall k \ge 1,
\eeq
for some constant $C>0,$ there exists at most one solution $\Gamma_t = (\gamma^{(k)}_t)_{k \ge 1}$ to \eqref{eq:GPhierarchy} with $\Gamma_{t=0} = \Gamma$ such that $||| \gamma^{(k)}_t |||_k \le C^k$ holds uniformly in $t.$

This motivates that one may begin directly with the GP hierarchy \eqref{eq:GPhierarchy} and investigates the associated initial problem on $\mathbb{R}^n.$ Recently, T. Chen and N. Pavlovi\'{c} \cite{CP2010} indeed started to investigate the Cauchy problem for the GP hierarchy \eqref{eq:GPhierarchy} without the assumption of factorized initial conditions. By introducing certain Sobolev type spaces $\mathcal{H}^{\alpha}_{\xi}$ of sequences of marginal density matrices with the parameter $0 < \xi <1,$ they proved local existence and uniqueness of solutions with an additional restriction on the interaction terms $(B^{(k)} \gamma^{(k+1)}_{t})_{k \ge 1}.$ Furthermore, Liu and the present author \cite{CL2011} proved the unconditional local wellposedness in the cas $\alpha > n/2.$ Their argument is also suitable to obtain the existence result in general case without assuming the additional condition used in \cite{CP2010} (see e.g. Theorem \ref{th:LocalWellposued-alpha>(n-1)/2} below). However, there appear two different parameters in these results, i.e., for the initial data in $\mathcal{H}^{\alpha}_{\xi_1}$ the solution lies in $\mathcal{H}^{\alpha}_{\xi_2}$ for some $0 < \xi_2 < \xi_1.$ In particular, for each fixed $\xi$ the norm $\| \cdot \|_{\mathcal{H}^{\alpha}_{\xi}}$ is not compatible with the Sobolev norm of $\mathrm{H}^{\alpha}$ for factorized hierarchies $\Gamma = \{\gamma^{(k)} \}_{k \ge 1}$ with
\beq\label{eq:FactorHiearchy}
\gamma^{(k)} = | \varphi \rangle \langle \varphi |^{\otimes^k},\quad k=1,2,\ldots, \quad \varphi \in \mathrm{H}^{\alpha}.
\eeq
This indicates that the Cauchy problem \eqref{eq:GPhierarchy} in $\mathcal{H}^{\alpha}_{\xi}$ is not equivalent to the one \eqref{eq:GPeq} in $\mathrm{H}^{\alpha}$ for the case that the initial condition is factorized.

In this paper we will eliminate these two undesirable issues. As a crucial ingredient of the arguments, a (F)-norm is introduced in certain Sobolev type spaces of sequences of marginal density matrices, i.e., a (F)-normed space $\mathcal{H}^{\alpha}$ consisting of all hierarchies $\Gamma = ( \gamma^{(k)} )_{k \ge 1}\in \bigotimes_{k=1}^{\infty} \mathrm{H}^{\alpha}_k$ such that
\be
\sum_{k=1}^{\infty} \frac{1}{\lambda^k} \|\gamma^{(k)}\|_{\mathrm{H}^{\alpha}_k} < \infty \;\text{for some}\; \lambda>0,
\ee
where $\mathrm{H}^{\alpha}_k = \mathrm{H}^{\alpha} (\mathbb{R}^{ k n} \times \mathbb{R}^{k n}),$ equipped with the (F)-norm
\be
\big \| \Gamma \big \|_{\mathcal{H}^{\alpha}} : = \inf \bigg \{ \lambda>0:\; \sum_{k=1}^{\infty} \frac{1}{\lambda^k} \|\gamma^{(k)}\|_{\mathrm{H}^{\alpha}_k} \le 1 \bigg \}.
\ee
Our main result in this paper is that the Cauchy problem of the GP hierarchy \eqref{eq:GPHierarchyEquaFunct} is locally well posed in $\mathcal{H}^{\alpha}$ for $\alpha > \max \{1/2, (n-1)/2\}.$

Note that the (F)-norm $\big \| \Gamma \big \|_{\mathcal{H}^{\alpha}}$ agrees with the condition \eqref{eq:ESYnorm} to a certain extent. Indeed, let $\mathfrak{H}^1$ be the space of sequences $\Gamma = (\gamma^{(k)})$ of trace class operators satisfying
\be
\| \Gamma \|_{\mathfrak{H}^1} : = \inf \bigg \{ \lambda>0:\; \sum_{k=1}^{\infty} \frac{1}{\lambda^k} |||\gamma^{(k)}|||_k \le 1 \bigg \} < \8.
\ee
It is easy to check that $\Gamma = (\gamma^{(k)})_{k \ge 1}$ satisfies \eqref{eq:ESYnorm} if and only if $\| \Gamma \|_{\mathfrak{H}^1} < \8.$ Then the uniqueness theorem of Erd\"{o}s, Schlein, and Yau stated above can be reformulated as follows: {\it The initial problem for the GP hierarchy \eqref{eq:GPhierarchy} admits at most one solution of density operators in $\mathfrak{H}^1.$} On the other hand, one can easily verify that $\| \Gamma \|_{\mathcal{H}^{\alpha}} = 2 \| \varphi \|^2_{\mathrm{H}^{\alpha}}$ for all factorized hierarchies $\Gamma$ of the form \eqref{eq:FactorHiearchy}, i.e., the (F)-norm $\| \cdot \|_{\mathcal{H}^{\alpha}}$ is compatible with $\|\cdot \|_{\mathrm{H}^{\alpha}}$ for any factorized hierarchy. This explains the reason why we introduce the parameter-free space $\mathcal{H}^{\alpha}$ in place of $\mathcal{H}^{\alpha}_{\xi}$ for studying the Cauchy problem of the GP hierarchy \eqref{eq:GPhierarchy}.

The paper is organized as follows. In Section \ref{PreMainresult}, some notations and the main results are presented. In Sections \ref{PfLocalWellposued1} we will prove unconditional local wellposedness in $\mathcal{H}^{\alpha}$ for the Cauchy problem of the hierarchy \eqref{eq:GPhierarchy} for any $\alpha > n/2.$ We further prove, in Section \ref{PfLocalWellposued2}, the local wellposedness with a priori assumption on  $(B^{(k)} \gamma^{(k+1)}_{t})_{k \ge 1}$ in both the stability and uniqueness parts. Instead of using a fixed point principle as in \cite{CP2010}, here we use the fully expanded iterated Duhamel series and a Cauchy convergence criterion involved in \cite{CL2011}. Finally, in Section \ref{GPHierarchyQuintic}, we will extend the result obtained for the cubic GP hierarchy to the so-called quintic one.

\section{Preliminaries and statement of the main result}\label{PreMainresult}

In what follows, we usually denote by $x= (x^1, \ldots, x^n)$ a general variable in $\mathbb{R}^n$ and by $\mathbf{x}_k=(x_1,\ldots, x_k)$ a point in $\mathbb{R}^{k n} = ( \mathbb{R}^n )^k.$ For any $x , y \in \mathbb{R}^n$ we denote by $x \cdot y = \sum^n_{i=1} x^i y^i$ and $|x|^2 = x \cdot x.$ For any $\mathbf{x}_k, \mathbf{y}_k \in \mathbb{R}^{k n},$ we set $\langle \mathbf{x}_k, \mathbf{y}_k \rangle: = \sum^k_{j=1} x_j \cdot y_j.$

We consider the Cauchy problem (initial value problem) for the Gross-Pitaevskii infinite linear
hierarchy of equations on $\mathbb{R}^n,$ in terms of kernel functions, of the form
\begin{equation}\label{eq:GPHierarchyEquaFunct}
\left \{ \begin{split}
& \big ( \mathrm{i} \partial_t + \triangle^{(k)}_{\pm} \big ) \gamma^{(k)}_t ({\bf x}_k;{\bf x}'_k) = \mu \big [ B^{(k)} (\gamma^{(k+1)}_t )\big ] ({\bf x}_k; {\bf x}'_k ),\\
& \gamma^{(k)}_{t=0}({\bf x}_k;{\bf x}'_k) = \gamma^{(k)}_{0}({\bf x}_k;{\bf x}'_k), \quad k =1, 2, \ldots, \end{split} \right.
\end{equation}
where $t \in \mathbb{R}, {\bf x}_k = (x_1, x_2, \ldots, x_k), {\bf x}'_k = (x'_1, x'_2, \ldots, x'_k) \in \mathbb{R}^{k n}, \mu= \pm 1.$ Here,
\be
\triangle^{(k)}_{\pm} = \sum^k_{j=1} ( \Delta_{x_j} - \Delta_{x'_j} )\quad \text{and}\quad B^{(k)} = \sum^k_{j=1} B_{j,k},
\ee
where $\Delta_{x_j}$ refers to the usual Laplace operator with respect to the variables $x_j \in {\mathbb R}^n$ and the operators $B_{j,k} = B_{j,k}^+ - B_{j,k}^-$ are defined according to
\begin{equation*}
\begin{split}
[ B_{j,k}^+ & (\gamma^{(k+1)})]  ({\bf x}_{k},{\bf x}'_{k})\\
= & \int \mathrm{d} x_{k+1} \mathrm{d}x'_{k+1}
\delta(x_{k+1}-x'_{k+1}) \delta(x_j- x_{k+1}) \gamma^{(k+1)}({\bf
x}_{k+1},{\bf x}'_{k+1}),
\end{split}
\end{equation*}
and
\begin{equation*}
\begin{split}
[B_{j,k}^- & (\gamma^{(k+1)}) ] ({\bf x}_{k},{\bf x}'_{k})\\
= & \int \mathrm{d} x_{k+1} \mathrm{d} x'_{k+1}
\delta(x_{k+1}-x'_{k+1} ) \delta(x'_j - x_{k+1})\gamma^{(k+1)}({\bf
x}_{k+1},{\bf x}'_{k+1}).
\end{split}
\end{equation*}
A sequence of functions $\Gamma (t) = ( \gamma^{(k)}_t )_{k \ge 1}$ is said to be a Gross-Pitaevskii (GP, in short) hierarchy, if they satisfy \eqref{eq:GPHierarchyEquaFunct} and are symmetric, in the sense that
\be
\gamma^{(k)}_t ({\bf x}_k, {\bf x}'_k) = \overline{\gamma^{(k)}_t ({\bf x}'_k, {\bf x}_k)}
\ee
and
\be
\gamma^{(k)}_t (x_1,\dotsc,x_k;x'_1,\dotsc,x'_k)= \gamma^{(k)}_t (x_{\sigma(1)},\dotsc,x_{\sigma(k)};x'_{\sigma(1)},\dotsc,x'_{\sigma(k)})
\ee
for any $\sigma \in \Pi_k$ ($\Pi_k$ denotes the set of permutations on $k$ elements).

Let $\varphi \in \mathrm{H}^1(\mathbb{R}^n),$ then one can easily verify that a particular GP hierarchy, i.e., a solution to \eqref{eq:GPHierarchyEquaFunct} with factorized initial datum
\be
\gamma^{(k)}_{t=0}({\bf x}_k; {\bf x}'_k) = \prod^k_{j=1} \varphi(x_j) \overline{\varphi(x'_j)},\quad k=1,2,\ldots,
\ee
is given by
\be
\gamma^{(k)}_t ({\bf x}_{k}; {\bf x}'_{k} ) = \prod^k_{j=1} \varphi_t (x_j) \overline{\varphi_t ( x'_j )},\quad  k=1,2,\ldots,
\ee
where $\varphi_t$ satisfies the cubic non-linear Schr\"odinger equation
\beq\label{eq:GPEqua}
\mathrm{i} \partial_t \varphi_t = -\Delta \varphi_t + \mu |\varphi_t|^2 \varphi_t,\quad \varphi_{t=0}=\varphi,
\eeq
which is {\it defocusing} if $\mu =1,$ and {\it focusing} if $\mu= -1.$ We refer to \cite{C2003} and references therein for the nonlinear Schr\"odinger equation.

In the following, unless otherwise specified, we always use $\gamma^{(k)}, \rho^{(k)}$ for denoting symmetric functions in $\mathbb{R}^{k n} \times \mathbb{R}^{k n}.$ For $k \geq 1$ and $\alpha \in \mathbb{R},$ we denote by $\mathrm{H}^{\alpha}_k = \mathrm{H}^{\alpha} (\mathbb{R}^{ k n} \times \mathbb{R}^{k n})$ the space of measurable functions $\gamma^{(k)} = \gamma^{(k)} ( {\bf x}_k, {\bf x}'_k )$ in $L^2(\mathbb{R}^{k n} \times \mathbb{R}^{k n}),$ which are symmetric, such that
\beq\begin{split}\label{normH}
\| \gamma^{(k)} \|_{\mathrm{H}^{\alpha}_k} : = \| S_k^{\alpha} \gamma^{(k)} \|_{L^2(\mathbb{R}^{k n} \times \mathbb{R}^{k n})} < \8,
\end{split}\eeq
where
\be\begin{split}
S_k^{\alpha} :  = \bigg ( \prod_{j=1}^k  (1 - \Delta_{x_j} )^{\frac{1}{2}} (1 - \Delta_{x'_j} )^{\frac{1}{2}} \bigg )^{\alpha}.
\end{split}\ee
Evidently, $\mathrm{H}^{\alpha}_k$ is a Hilbert space with the inner product
\be\begin{split}
\langle \gamma^{(k)}, \rho^{(k)} \rangle = \big \langle S_k^{\alpha} \gamma^{(k)},\; S_k^{\alpha} \rho^{(k)} \big \rangle_{L^2(\mathbb{R}^{k n} \times \mathbb{R}^{k n})}.
\end{split}\ee
Moreover, the norm $\|\cdot\|_{\mathrm{H}^{\alpha}_k}$ is invariance under the action of $e^{\mathrm{i} t \triangle^{(k)}},$ i.e.,
\be
\| e^{\mathrm{i} t \triangle^{(k)}} \gamma^{(k)} \|_{\mathrm{H}^{\alpha}_k} = \| \gamma^{(k)} \|_{\mathrm{H}^{\alpha}_k}
\ee
because $e^{\mathrm{i} t \triangle^{(k)}}$ commutates with $\Delta_{x_j}$ for any $j.$

\begin{definition}\label{df:SobolevSpace}
For $\alpha \in \mathbb{R}$ we define
\be
\mathcal{H}^{\alpha} = \bigg \{ ( \gamma^{(k)} )_{k \ge 1} \in \bigotimes_{k=1}^{\infty} \mathrm{H}^{\alpha}_k :\; \sum_{k=1}^{\infty} \frac{1}{\lambda^k} \|\gamma^{(k)}\|_{\mathrm{H}^{\alpha}_k} < \infty \;\text{for some}\; \lambda > 0 \bigg \},
\ee
equipped with the $\mathrm{(F)}$-norm
\beq\label{eq:SobolevSpaceNorm}
\big \| ( \gamma^{(k)} )_{k \ge 1} \big \|_{\mathcal{H}^{\alpha}} : = \inf \bigg \{ \lambda > 0:\; \sum_{k=1}^{\infty} \frac{1}{\lambda^k} \|\gamma^{(k)}\|_{\mathrm{H}^{\alpha}_k} \le 1 \bigg \}.
\eeq
\end{definition}

\begin{remark}\label{rk:SobolevSpace}\rm
Note that $\| \cdot \|_{\mathcal{H}^{\alpha}}$ is not actually a norm in the sense that it does not satisfy $\| \lambda \Gamma \|_{\mathcal{H}^{\alpha}} = |\lambda| \| \Gamma \|_{\mathcal{H}^{\alpha}}$ in general. However, it is a (F)-norm, that is,
\begin{enumerate}[{\rm (F1)}]

\item $\| \Gamma \|_{\mathcal{H}^{\alpha}} \ge 0;$

\item $\| \Gamma \|_{\mathcal{H}^{\alpha}} = 0$ if and only if $\Gamma =0;$

\item $\| \lambda \Gamma \|_{\mathcal{H}^{\alpha}} \le \| \Gamma \|_{\mathcal{H}^{\alpha}}$ if $| \lambda | \le 1;$

\item $\|\Gamma + \Theta \|_{\mathcal{H}^{\alpha}} \le \|\Gamma \|_{\mathcal{H}^{\alpha}}+ \| \Theta \|_{\mathcal{H}^{\alpha}};$

\item $\| \lambda \Gamma_n \|_{\mathcal{H}^{\alpha}} \to 0$ if $\| \Gamma_n \|_{\mathcal{H}^{\alpha}} \to 0;$

\item $\| \lambda_n \Gamma \|_{\mathcal{H}^{\alpha}} \to 0$ if $\lambda_n \to 0.$
\end{enumerate}
The conditions (F1)-(F3) are clear from the definition. (F4) results from the fact that $\mathrm{H}^{\alpha}_k$'s are all norms, while (F5) and (F6) follow easily from Lebesgue's dominated convergence theorem and (F1)-(F4). Moreover, it can be checked that $\mathcal{H}^{\alpha}$ is complete under the metric determined by this (F)-norm. Therefore, $\mathcal{H}^{\alpha}$ is a complete (F)-normed space. We refer to \cite[\S 15.11]{Kothe1983} for a detailed account of (F)-norms.
\end{remark}

\begin{definition}\label{df:SpacetimeSobolevSpace}
For an interval $I \subset \mathbb{R},$ we define $L_{t\in I}^1\mathcal{H}^{\alpha}$ to be the space of all strongly measurable functions $\Gamma(t)=\{\gamma^{(k)}_t\}_{k\geq 1}$ on $I$ with values in $\mathcal{H}^{\alpha}$ such that
\begin{equation*}
\sum_{k=1}^{\infty} \frac{1}{\lambda^{k}}\int_I \|\gamma^{(k)}_t\|_{\mathrm{H}_k^{\alpha}}dt<\infty\quad \text{for some} \ \lambda>0,
\end{equation*}
equipped with the $\mathrm{(F)}$-norm
\begin{equation}\label{ndef2}
\|\Gamma(t)\|_{L_{t \in I}^1 \mathcal{H}^{\alpha}}: = \inf \bigg \{ \lambda>0:  \sum_{k=1}^{\infty} \frac{1}{\lambda^{k}}\int_I \|\gamma^{(k)}_t\|_{\mathrm{H}_k^{\alpha}} d t \leq 1 \bigg \}.
\end{equation}
\end{definition}

\begin{remark}\label{rk:SobolevtimeSpace}\rm
It follows easily from the Lebesgue dominated convergence theorem that $L_{t\in I}^1\mathcal{H}^{\alpha}$ is a complete (F)-normed space. See Remark \ref{rk:SobolevSpace} above.
\end{remark}

Recall that, in integral formulation, \eqref{eq:GPHierarchyEquaFunct} can be written as
\beq\label{eq:GPHierarchyFunctIntEqua}
\gamma^{(k)}_t = e^{\mathrm{i} t \triangle^{(k)}} \gamma^{(k)}_0 + \int^{t}_{0} d s\;  e^{\mathrm{i} (t-s) \triangle^{(k)}} \tilde{B}^{(k)} \gamma^{(k+1)}_s,\; k=1,2,\ldots\;,
\eeq
whereafter $\tilde{B}^{(k)} = - \mathrm{i} \mu B^{(k)}.$ As noted in \cite{CP2010}, such a solution can be obtained by solving the following infinite linear hierarchy of integral equations
\beq\label{eq:GPStrongSolutionDuh-kth}\begin{split}
\tilde{B}^{(k)} \gamma^{(k+1)}_t & = \tilde{B}^{(k)} e^{\mathrm{i} t \triangle^{(k+1)}} \gamma^{(k+1)}_0 + \int^{t}_{0} d s\, \tilde{B}^{(k)} e^{\mathrm{i} (t-s) \triangle^{(k+1)}} \tilde{B}^{(k+1)} \gamma^{(k+2)}_s,
\end{split}\eeq
for any $k \ge 1.$ If we write
\be
\hat{\triangle} \Gamma : = \big ( \triangle^{(k)} \gamma^{(k)} \big )_{k \ge 1} \quad \text{and} \quad \hat{B} \Gamma := \big ( \tilde{B}^{(k)} \gamma^{(k+1)} \big )_{k \ge 1},
\ee
then \eqref{eq:GPHierarchyFunctIntEqua} and \eqref{eq:GPStrongSolutionDuh-kth} can be written as
\beq\label{eq:GPHierarchyFunctIntEquaGamma}
\Gamma (t) = e^{\mathrm{i} t \hat{\triangle}} \Gamma_0 + \int^{t}_{0} d s~  e^{\mathrm{i} (t-s) \hat{\triangle}} \hat{B} \Gamma (s)
\eeq
and
\beq\label{eq:GPStrongSolutionFunctGamma}\begin{split}
\hat{B} \Gamma (t) = \hat{B} e^{\mathrm{i} t \hat{\triangle}} \Gamma_0 + \int^{t}_{0} d s\, \hat{B} e^{\mathrm{i} (t-s) \hat{\triangle}} \hat{B} \Gamma (s),
\end{split}\eeq
respectively.

Let us make the notion of solution more precise.

\begin{definition}\label{df:StrongSolution}
A function $\Gamma (t) = ( \gamma^{(k)}_t )_{k \geq 1}: I \mapsto \mathcal{H}^{\alpha}$ on a non-empty time interval $0 \in I \subset \mathbb{R}$ is said to be a local $(strong)$ solution
to the Gross-Pitaevskii hierarchy \eqref{eq:GPHierarchyEquaFunct} if it lies in the class $C (K, \mathcal{H}^{\alpha})$ for all compact sets $K \subset I$ and obeys the Duhamel formula
\beq\label{eq:MildSolution}
\gamma^{(k)}_t = e^{\mathrm{i} t \triangle^{(k)}} \gamma^{(k)}_0 - \mathrm{i} \mu \int^{t}_{0} d s\; e^{\mathrm{i} (t-s) \triangle^{(k)}} B^{(k)} \gamma^{(k+1)}_s,\quad \forall t \in I,
\eeq
holds in $\mathrm{H}^{\alpha}_k$ for every $k=1,2,\ldots.$
\end{definition}

We refer to the interval $I$ as the {\it lifespan} of $\Gamma (t)$ with the initial value $\Gamma (0) = ( \gamma^{(k)}_0 )_{k \geq 1}.$ We say that $\Gamma (t)$ is a {\it maximal-lifespan solution} on $I = (- T_{\mathrm{min}}, T_{\mathrm{max}})$ with $T_{\mathrm{max}} = T_{\mathrm{max}} (\Gamma_0) \in (0, \8]$ and $T_{\mathrm{min}} = T_{\mathrm{min}} (\Gamma_0) \in (0, \8]$ if the solution cannot be extended to any strictly larger interval. $(- T_{\mathrm{min}}, T_{\mathrm{max}})$ is said to be the {\it maximal lifespan} of $\Gamma (t)$ with the initial value $\Gamma (0) = ( \gamma^{(k)}_0 )_{k \geq 1}.$ We say that $\Gamma (t)$ is a {\it global} solution if $(- T_{\mathrm{min}}, T_{\mathrm{max}}) = \mathbb{R},$ i.e., $T_{\mathrm{max}} = T_{\mathrm{min}} = \8.$

\begin{remark}\label{rk:Factoriedcase}\rm
Let $\varphi \in \mathrm{H}^{\alpha}(\mathbb{R}^n)$ and set for $k \ge 1,$
\be
\gamma^{(k)}_0 ({\bf x}_k; {\bf x}'_k) = \prod^k_{j=1} \varphi(x_j) \overline{\varphi(x'_j)}.
\ee
An immediate computation yields that
\be
\big \| ( \gamma^{(k)}_0 )_{k \ge 1} \big \|_{\mathcal{H}^{\alpha}} = 2 \| \varphi \|^2_{\mathrm{H}^{\alpha} (\mathbb{R}^n)}.
\ee
Thus, for $T>0,$ $\varphi_t \in C((-T,T), \mathrm{H}^{\alpha})$ is a solution to \eqref{eq:GPEqua} with the initial value $\varphi_t |_{t=0} = \varphi$ if and only if \beq\label{eq:FactorGPHierarchy}
\Gamma (t) = ( \gamma^{(k)}_t )_{k \ge 1} \; \text{with}\; \gamma^{(k)}_t({\bf x}_k; {\bf x}'_k) = \prod^k_{j=1} \varphi_t(x_j) \overline{\varphi_t(x'_j)}
\eeq
is a solution to \eqref{eq:GPHierarchyEquaFunct} in $C((-T,T), \mathcal{H}^{\alpha})$ with $\Gamma (0) = ( \gamma^{(k)}_0 )_{k \ge 1}.$
This yields that the Cauchy problem \eqref{eq:GPHierarchyEquaFunct} in $\mathcal{H}^{\alpha}$ is equivalent to the one \eqref{eq:GPEqua}
in $\mathrm{H}^{\alpha}$ for the case of that initial conditions are factorized.
\end{remark}

The aim of this paper is to prove the local well-posedness of the GP hierarchy \eqref{eq:GPHierarchyEquaFunct} in $\mathcal{H}^{\alpha}$ for $\alpha > \max \{1/2,\; (n-1)/2 \}.$ We state our results as Theorems \ref{th:LocalWellposued-alpha>n/2} and \ref{th:LocalWellposued-alpha>(n-1)/2}.

\begin{theorem}\label{th:LocalWellposued-alpha>n/2}
Assume that $n \geq 1$ and $\alpha > n/2.$ The Cauchy problem \eqref{eq:GPHierarchyEquaFunct} is locally well posed. More precisely, there exists a constant $A_{n, \alpha}>0$ depending only on $n$ and $\alpha$ such that
\begin{enumerate}[{\rm (i)}]

\item For each $\Gamma_0 = ( \gamma^{(k)}_{0} )_{k \geq 1}\in \mathcal{H}^{\alpha},$ let $I = [-T, T]$ with $T =\frac{A_{n, \alpha}}{\| \Gamma_0 \|_{\mathcal{H}^{\alpha}}}.$ Then there exists a solution $\Gamma (t) = ( \gamma^{(k)}_t )_{k \geq 1} \in C ( I, \mathcal{H}^{\alpha})$ to the Gross-Pitaevskii hierarchy \eqref{eq:GPHierarchyEquaFunct} with the initial data $\Gamma_0$ such that
\begin{equation}\label{eq:SpacetimeEstimate-LocalWellposued1}
\| \Gamma (t) \|_{C( I,\mathcal{H}^{\alpha})} \leq 2 \| \Gamma_0 \|_{\mathcal{H}^{\alpha}}.
\end{equation}

\item Given $I_0 = [-T_0, T_0]$ with $T_0 > 0,$ if $\Gamma (t)$ and $\Gamma' (t)$ in $C ( I_0, \mathcal{H}^{\alpha})$ are two solutions to \eqref{eq:GPHierarchyEquaFunct} with initial conditions $\Gamma_{t=0} = \Gamma_0$ and $\Gamma'_{t=0} = \Gamma'_0$ in $\mathcal{H}^{\alpha}$ respectively, then
\begin{equation}\label{eq:SpacetimeEstimate-Stability}
\| \Gamma (t) - \Gamma' (t) \|_{C( I, \mathcal{H}^{\alpha})} \leq 2 \| \Gamma_0 - \Gamma'_0  \|_{\mathcal{H}^{\alpha}},
\end{equation}
with $I = [-T, T],$ where
\be
T = \min \left \{T_0, \frac{A_{n, \alpha}}{\| \Gamma (t) - \Gamma' (t) \|_{C ( I_0; \mathcal{H}^{\alpha})}} \right \}.
\ee
\end{enumerate}

\end{theorem}

\begin{remark}\label{rk:thalpha>n/2}\rm
Theorem \ref{th:LocalWellposued-alpha>n/2} shows that one has unconditional local wellposedness in $\mathcal{H}^{\alpha}$ for the Cauchy problem of the GP hierarchy \eqref{eq:GPHierarchyEquaFunct} for any $\alpha > n/2.$ This agrees with the case of the GP equation \eqref{eq:GPEqua} (cf. \cite[Proposition 3.8]{Tao2006}). Using the classical argument (see the proof of \cite[Theorem 1.17]{Tao2006}), one can easily shows that given any $\Gamma_0 \in \mathcal{H}^{\alpha},$ there exists a unique maximal interval of existence $I$ and a unique solution $\Gamma_t \in C (I, \mathcal{H}^{\alpha}).$ Moreover, if $I$ has a finite endpoint $T,$ i.e., $T = T_{\max} < \8$ or $T = T_{\min} <\8,$ then the $\mathcal{H}^{\alpha}$-norm of $\Gamma_t$ will go to infinity as $t \to T.$ Thus, the maximal lifespan of $\Gamma_t$ is necessarily open.
\end{remark}

\begin{theorem}\label{th:LocalWellposued-alpha>(n-1)/2}
Assume that $n \ge 2$ and $\alpha > (n-1)/2.$ Then, the Cauchy problem for the Gross-Pitaevskii hierarchy \eqref{eq:GPHierarchyEquaFunct} is locally well posed in $\mathcal{H}^{\alpha}.$ More precisely, there exist an absolute constant $A>2$ and a constant $C=B_{n, \alpha}>0$ depending only on $n$ and $\alpha$ such that
\begin{enumerate}[{\rm (1)}]

\item For every $\Gamma_0 = ( \gamma^{(k)}_{0} )_{k \geq 1} \in \mathcal{H}^{\alpha},$ let $I = [-T, T]$ with $T = \frac{B_{n, \alpha}}{\| \Gamma_0 \|^2_{\mathcal{H}^{\alpha}}}.$ Then there exists a solution $\Gamma (t) = ( \gamma^{(k)} (t) )_{k \geq 1} \in C ( I, \mathcal{H}^{\alpha})$ to \eqref{eq:GPHierarchyEquaFunct} with the initial data $\Gamma (0) = \Gamma_0$ satisfying
\begin{equation}\label{eq:SpacetimeEstimate}
\| \hat{B} \Gamma (t) \|_{L^1_{t \in I} \mathcal{H}^{\alpha}} \leq 4 A \| \Gamma_0 \|_{\mathcal{H}^{\alpha}}.
\end{equation}

\item Given $I_0 = [-T_0, T_0]$ with $T_0 >0,$ if $\Gamma (t) \in C ( I_0, \mathcal{H}^{\alpha})$ so that $\hat{B} \Gamma (t) \in L^1_{t \in I_0} \mathcal{H}^{\alpha}$ is a solution to \eqref{eq:GPHierarchyEquaFunct} with the initial data $\Gamma (0) = \Gamma_0,$ then \eqref{eq:SpacetimeEstimate} holds as well for $I = [-T, T],$ where
\be
T = \min \bigg \{T_0, \; \frac{B_{n, \alpha}}{ \| \hat{B} \Gamma (t) \|^2_{L^1_{t \in I_0} \mathcal{H}^{\alpha}} + \| \Gamma_0 \|^2_{\mathcal{H}^{\alpha}} } \bigg \}.
\ee

\item Given $I_0 = [-T_0, T_0]$ with $T_0 >0,$ if $\Gamma (t)$ and $\Gamma' (t)$ in $C ( I_0, \mathcal{H}^{\alpha})$ with $\hat{B} \Gamma (t), \hat{B} \Gamma' (t) \in L^1_{t \in I_0}\mathcal{H}^{\alpha}$ are two solutions to \eqref{eq:GPHierarchyEquaFunct} with initial conditions $\Gamma (0) = \Gamma_0$ and $\Gamma' (0) = \Gamma'_0$ in $\mathcal{H}^{\alpha}$ respectively, then
\begin{equation}\label{eq:InitialValueContinuousDependence}
\| \Gamma (t) - \Gamma' (t) \|_{C( I, \mathcal{H}^{\alpha})} \leq (1+ 4 A) \| \Gamma_0 - \Gamma'_0\|_{\mathcal{H}^{\alpha}},
\end{equation}
with $I =[-T, T],$ where
\be
T = \min \bigg \{ T_0, \; \frac{B_{n, \alpha}}{ \| \hat{B} [\Gamma (t) - \Gamma' (t) ]\|^2_{L^1_{t \in I_0} \mathcal{H}^{\alpha}} + \| \Gamma_0 - \Gamma'_0 \|^2_{\mathcal{H}^{\alpha}} } \bigg \}.
\ee

\end{enumerate}

In particular, the above results hold for $\mathcal{H}^1$ in the case $n=3.$
\end{theorem}

\begin{remark}\label{rk:thalpha>(n-1)/2}\rm
As shown in Theorem \ref{th:LocalWellposued-alpha>(n-1)/2}, for the case $n/2 \ge \alpha > (n-1)/2$ we require a priori assumption $\hat{B} \Gamma (t) \in L^1_{t \in I} \mathcal{H}^{\alpha}$ in both the stability and uniqueness parts, although we prove that for the existence part,
such a priori assumption is not required. At the time of this writing, the question remains open whether the
condition $\hat{B} \Gamma (t) \in L^1_{t \in I} \mathcal{H}^{\alpha}$ is necessary for the uniqueness of solutions.
\end{remark}

\section{Proof of Theorem \ref{th:LocalWellposued-alpha>n/2}}\label{PfLocalWellposued1}

We begin with the following lemma.

\begin{lemma}\label{le:BOperatorEstimate-LocalWellposued1}
Suppose that $n\geq 1$ and $\alpha>\dfrac{n}{2}.$ Then, there exists a constant $ C_{n, \alpha} >0$ depending only on $n$ and $\alpha$ such that, for any $\gamma^{(k+1)} \in \mathcal{S} (\mathbb{R}^{(k+1) n} \times \mathbb{R}^{(k+1) n}),$
\be
\| B_{j,k} \gamma^{(k+1)} \|_{\mathrm{H}^{\alpha}_k} \leq C_{n, \alpha} \|\gamma^{(k+1)} \|_{\mathrm{H}^{\alpha}_{k+1}},
\ee
for all $k \geq 1,$ where $j=1,\cdots,k.$

Consequently, $B^{(k)}$ can be extended to a bounded operator from $\mathrm{H}^{\alpha}_{k+1}$ to $\mathrm{H}^{\alpha}_k,$ still denoted by $B^{(k)},$ satisfying
\beq\label{eq:BoperatorEstimate-LocalWellposued1}
\|B^{(k)} \gamma^{(k+1)} \|_{\mathrm{H}^{\alpha}_k} \leq C_{n, \alpha} k \|\gamma^{(k+1)}\|_{\mathrm{H}^{\alpha}_{k+1}}.
\eeq
for all $\gamma^{(k+1)} \in \mathrm{H}^{\alpha}_{k+1}.$
\end{lemma}

This result can be found in \cite{CP2010} and \cite{CL2011}. The proof of Theorem \ref{th:LocalWellposued-alpha>n/2} is divided into two parts as follows. Without loss of generality, we always assume $t \ge 0.$

\begin{proof}
(i)\; Let $\alpha > n/2.$ Given $\Gamma_0 = \{ \gamma^{(k)}_0 \}_{k \ge 1} \in \mathcal{H}^{\alpha}.$ For $m \geq 1,$ set
\begin{equation}\label{eq:m-iterate}
\gamma^{(k)}_{m,t} = e^{\mathrm{i} t \triangle^{(k)}_{\pm}} \gamma^{(k)}_0 + \int^{t}_{0} \mathrm{d} s\; e^{\mathrm{i} (t - s) \triangle^{(k)}_{\pm}} \tilde{B}^{(k)} \gamma^{(k+1)}_{m-1,s},\quad k \geq 1,
\end{equation}
with the convention $\gamma^{(k)}_{0,t} \equiv \gamma^{(k)}_{0}.$ By expansion, for every $m \geq 1$ one has
\begin{equation*}
\begin{split}
\gamma^{(k)}_{m,t}= & e^{\mathrm{i} t \triangle^{(k)}_{\pm}} \gamma^{(k)}_{0} +
\sum_{j=1}^{m-1}\int_{0}^{t}\mathrm{d} t_1 \int_0^{t_1}\mathrm{d}t_2
\cdots \int_0^{t_{j-1}}\mathrm{d}
t_j e^{\mathrm{i} (t-t_1) \triangle^{(k)}_{\pm}} \tilde{B}^{(k)} \cdots\\
& \times e^{\mathrm{i} (t_{j-1} - t_j ) \triangle^{(k+j-1)}_{\pm}} \tilde{B}^{(k+j-1)}
e^{\mathrm{i} t_j \triangle^{(k+j)}_{\pm}} \gamma^{(k+j)}_0\\
& + \int_{0}^{t} \mathrm{d} t_1 \int_0^{t_1} \mathrm{d} t_2 \cdots
\int_0^{t_{m-1}} \mathrm{d} t_m e^{\mathrm{i} (t-t_1) \triangle^{(k)}_{\pm}}
\tilde{B}^{(k)} \cdots\\
&\times e^{\mathrm{i} (t_{m-1} - t_m ) \triangle^{(k+j-1)}_{\pm}} \tilde{B}^{(k+m-1)}
\gamma^{(k+m)}_{0}
\triangleq  \sum_{j=0}^{m} \Xi_{j,t}^{(k)},
\end{split}
\end{equation*}
with the convention $t_0 =t.$ Then, for $j=1,\cdots,m-1$, by Lemma \ref{le:BOperatorEstimate-LocalWellposued1} we
have
\beq\label{eq:XiEsitmate}
\begin{split}
\| \Xi_{j,t}^{(k)} \|_{\mathrm{H}^{\alpha}_k} \leq & \int_0^t \mathrm{d} t_1 \int_0^{t_1} \mathrm{d} t_2 \cdots \int_0^{t_{j-1}}
\mathrm{d} t_j \Big \|e^{\mathrm{i} (t-t_1) \triangle^{(k)}_{\pm}} \tilde{B}^{(k)} \cdots \\
& \times e^{\mathrm{i} (t_{j-1} - t_j ) \triangle^{(k+j-1)}_{\pm}} \tilde{B}^{(k+j-1)} e^{\mathrm{i} t_j \triangle^{(k+j)}_{\pm}} \gamma^{(k+j)}_{0} \Big \|_{\mathrm{H}^{\alpha}_k}\\
\leq & \int_{0}^{t} \mathrm{d} t_1 \int_0^{t_1} \mathrm{d}t_2 \cdots \int_0^{t_{j-1}} \mathrm{d} t_j k \cdots \\
& \times (k+j-1)(C_{n, \alpha})^{j} \| e^{\mathrm{i} t_j \triangle^{(k+j)}_{\pm}} \gamma^{(k+j)}_{0}\|_{\mathrm{H}^{\alpha}_{k+j}}\\
= & \dfrac{t^j}{j!} k \cdots (k+j-1) (C_{n, \alpha})^j \|\gamma^{(k+j)}_{0}\|_{\mathrm{H}^{\alpha}_{k+j}}\\
= & \binom{k+j-1}{j} (C_{n, \alpha}  t )^j\|\gamma^{(k+j)}_{0}\|_{\mathrm{H}^{\alpha}_{k+j}},
\end{split}
\eeq
and
\begin{equation*}
\begin{split}
\| \Xi_{m,t}^{(k)}\|_{\mathrm{H}^{\alpha}_k} \leq & \int_{0}^{t} \mathrm{d} t_1 \int_0^{t_1} \mathrm{d} t_2 \cdots \int_0^{t_{m-1}} \mathrm{d} t_m \Big \| e^{\mathrm{i} (t-t_1) \triangle^{(k)}_{\pm}} \tilde{B}^{(k)} \cdots \\
& \times e^{\mathrm{i} (t_{m-1} - t_m ) \triangle^{(k+j-1)}_{\pm}} \tilde{B}^{(k+m-1)} \gamma^{(k+m)}_0 \Big \|_{\mathrm{H}^{\alpha}_k}\\
\leq & \int_{0}^{t}\mathrm{d}t_1\int_0^{t_1}\mathrm{d}t_2\cdots\int_0^{t_{m-1}}\mathrm{d} t_m~ k \cdots \\
& \times (k+m-1) (C_{n, \alpha})^{m} \| \gamma^{(k+m)}_0 \|_{\mathrm{H}^{\alpha}_{k+m}}\\
\le & \dfrac{t^m}{m!} k \cdots (k+m-1)(C_{n, \alpha})^{m} \|\gamma^{(k+m)}_0\|_{\mathrm{H}^{\alpha}_{k+m}}\\
= & \binom{k+ m-1}{m} (C_{n, \alpha}  t )^m \|\gamma^{(k+m)}_0 \|_{\mathrm{H}^{\alpha}_{k+m}}.
\end{split}
\end{equation*}
Then, for $T > 0$ ($T$ will be fixed in the sequel) we obtain
\be
\begin{split}
\|\gamma^{(k)}_{m,t} \|_{C([0,T], \mathrm{H}^{\alpha}_k)} & \le \sum_{j=0}^{m} \|\Xi_{j}^{(k)}\|_{C([0,T], \mathrm{H}^{\alpha}_k)} \leq  \sum_{j = 0}^m \binom{k+j-1}{j} (C_{n, \alpha}  T )^j\|\gamma^{(k+j)}_{0}\|_{\mathrm{H}^{\alpha}_{k+j}}.
\end{split}
\ee
Hence, for $\lambda > 0$ one has
\begin{equation}\label{eq:GammaSum}
\begin{split}
\sum_{k=1}^{\8} \frac{1}{\lambda^k} \|\gamma^{(k)}_{m,t} \|_{C([0,T], \mathrm{H}^{\alpha}_k)}
& \leq  \sum_{j = 0}^m \sum_{k=1}^{\8} \frac{1}{\lambda^k} \binom{k+j-1}{j} (C_{n, \alpha}  T )^j\|\gamma^{(k+j)}_{0}\|_{\mathrm{H}^{\alpha}_{k+j}}\\
& \leq \sum_{j = 0}^{\8} \sum_{k=1}^{\8} \frac{1}{\lambda^k} \binom{k+j-1}{j} (C_{n, \alpha}  T )^j\|\gamma^{(k+j)}_{0}\|_{\mathrm{H}^{\alpha}_{k+j}}.\\
\end{split}
\end{equation}
By the direct computation, one has
\be
\begin{split}
\sum_{j=0}^{\infty}\sum_{k=1}^{\infty} & \frac{1}{\lambda^k} \binom{k+j-1}{j} (C_{n, \alpha}  T )^j\|\gamma^{(k+j)}_{0}\|_{\mathrm{H}^{\alpha}_{k+j}}\\
=&\sum_{j=0}^{\infty}\sum_{\el=j+1}^{\infty} \frac{1}{\lambda^{\el - j}}\binom{\el-1}{j}(C_{n, \alpha}  T )^j\|\gamma^{(\el)}_{0}\|_{\mathrm{H}^{\alpha}_{\el}}\\
=&\sum_{\el=1}^{\infty}\sum_{j=0}^{\el-1}\binom{\el-1}{j}(C_{n, \alpha}  T \lambda )^j \frac{1}{\lambda^{\el}} \|\gamma^{(\el)}_{0}\|_{\mathrm{H}^{\alpha}_{\el}}\\
=&\sum_{\el=1}^{\infty}(1+C_{n, \alpha}  T \lambda )^{\el-1} \frac{1}{\lambda^{\el}}\|\gamma^{(\el)}_{0}\|_{\mathrm{H}^{\alpha}_{\el}}
\le \sum_{\el=1}^{\infty}\Big ( \frac{1}{\lambda} + C_{n, \alpha} T \Big )^{\el} \|\gamma^{(\el)}_{0}\|_{\mathrm{H}^{\alpha}_{\el}}.
\end{split}
\ee
Set $\Gamma_m (t) = \{ \gamma^{(k)}_{m,t} \}.$ Let $T =1 / [4 C_{n, \alpha} \| \Gamma_0 \|_{\mathcal{H}^{\alpha}}].$ Then, by choosing $\lambda = 4 \|\Gamma_0\|_ {\mathcal{H}^{\alpha}}$ we conclude from \eqref{eq:GammaSum} that
\beq\label{eq:SpacetimeEstimate_m-itera}
\| \Gamma_m (t) \|_{C( [0,T], \mathcal{H}^{\alpha})} \leq 2 \|\Gamma_0\|_ {\mathcal{H}^{\alpha}}.
\eeq

Now, fix $k \ge 1,$ by the above estimates for $\Xi_{j,t}^{(k)}$ we have for any $m,n$ with $n >m \gg k$
\be\begin{split}
\|\gamma^{(k)}_{m,t} - \gamma^{(k)}_{n,t}\|_{C([0,T], \mathrm{H}^{\alpha}_k)} & \le 2 \sum_{j = m}^{n} \binom{k+j-1}{j} (C_{n, \alpha}  T )^j\|\gamma^{(k+j)}_{0}\|_{\mathrm{H}^{\alpha}_{k+j}}\\
& \le C_k \sum_{j = m}^{n} \frac{(k+j-1)^{k- 1/2}}{\sqrt{j}} (C_{n, \alpha}  T )^j\|\gamma^{(k+j)}_{0}\|_{\mathrm{H}^{\alpha}_{k+j}}\\
& = C_k \sum_{j = m}^{n} \frac{(k+j-1)^{k- 1/2}}{2^j \sqrt{j}} \frac{1}{(2 \| \Gamma_0 \|_{\mathcal{H}^{\alpha}})^j} \|\gamma^{(k+j)}_{0}\|_{\mathrm{H}^{\alpha}_{k+j}}\\
& \le C_k \frac{(2 \| \Gamma_0 \|_{\mathcal{H}^{\alpha}})^k}{\sqrt{m}} \sum_{j = m}^{n} \frac{1}{( 2 \| \Gamma_0 \|_{\mathcal{H}^{\alpha}})^{k+j}} \|\gamma^{(k+j)}_{0}\|_{\mathrm{H}^{\alpha}_{k+j}} \to 0,
\end{split}\ee
as $m \to \8,$ where we have used Stirling's formula $N! \approx N^{N + 1/2} e^{-N}.$ This concludes that for every $k \ge 1,$ $\gamma^{(k)}_{m,t}$ converges in $C([0,T], \mathrm{H}^{\alpha}_k),$ whose limitation is denoted by $\gamma^{(k)}_t.$

Set $\Gamma (t) = \{ \gamma^{(k)}_t \}_{k \ge 1}.$ By Lemma \ref{le:BOperatorEstimate-LocalWellposued1} one has
\be\begin{split}
\Big \| \int^{t}_{0} \mathrm{d} s\; & e^{\mathrm{i} (t - s) \triangle^{(k)}_{\pm}} \tilde{B}^{(k)} [ \gamma^{(k+1)}_{m-1,s} - \gamma^{(k+1)}_{n-1,s} ]\Big \|_{\mathrm{H}^{\alpha}_k }\\
& \le \int^T_{0} \mathrm{d} s \| \tilde{B}^{(k)} [ \gamma^{(k+1)}_{m-1,s} - \gamma^{(k+1)}_{n-1,s} ] \|_{\mathrm{H}^{\alpha}_k }
\le k C_{n, \alpha} T \| \gamma^{(k+1)}_{m-1,t} - \gamma^{(k+1)}_{n-1,t} \|_{C([0,T], \mathrm{H}^{\alpha}_{k+1})}.
\end{split}\ee
Then, taking $m \to \8$ in \eqref{eq:m-iterate} we conclude that $\Gamma (t)$ is a solution to \eqref{eq:GPHierarchyEquaFunct}. Moreover, taking $m \to \8$ in \eqref{eq:SpacetimeEstimate_m-itera} we obtain \eqref{eq:SpacetimeEstimate-LocalWellposued1}.

(ii)\; Given $T_0 >0.$ Suppose $\Gamma (t), \Gamma' (t) \in C( [0,T_0], \mathcal{H}^{\alpha})$ are two solutions to \eqref{eq:GPHierarchyEquaFunct} with the initial datum $\Gamma_0$ and $\Gamma'_0$ in $\mathcal{H}^{\alpha},$ respectively. Since \eqref{eq:GPHierarchyEquaFunct} is linear, it suffices to consider $\Gamma (t)$ instead of $\Gamma (t)-\Gamma' (t).$
By \eqref{eq:GPHierarchyFunctIntEqua}, for every $m \geq 1$ one has
\begin{equation*}
\begin{split}
\gamma^{(k)}_t = & e^{\mathrm{i} t \triangle^{(k)}_{\pm}} \gamma^{(k)}_{0} +
\sum_{j=1}^{m-1}\int_{0}^{t}\mathrm{d}t_1\int_0^{t_1}\mathrm{d}t_2
\cdots \int_0^{t_{j-1}}\mathrm{d}
t_j e^{\mathrm{i} (t - t_1) \triangle^{(k)}_{\pm}} \tilde{B}^{(k)} \cdots\\
& \times e^{\mathrm{i} (t_{j-1} - t_j) \triangle^{(k+j-1)}_{\pm}} \tilde{B}^{(k+j-1)}
e^{\mathrm{i} t_j \triangle^{(k+j)}_{\pm}} \gamma^{(k+j)}_0\\
& + \int_{0}^{t} \mathrm{d} t_1 \int_0^{t_1} \mathrm{d} t_2 \cdots
\int_0^{t_{m-1}} \mathrm{d} t_m e^{\mathrm{i} (t - t_1) \triangle^{(k)}_{\pm}}
\tilde{B}^{(k)} \cdots\\
&\times e^{\mathrm{i} (t_{m-1} - t_m) \triangle^{(k+m-1)}_{\pm}} \tilde{B}^{(k+m-1)}
\gamma^{(k+m)}_{t_m}
\triangleq \sum_{j=0}^{m-1} \Xi_{j,t}^{(k)} + \tilde{\Xi}_{m,t}^{(k)},
\end{split}
\end{equation*}
with the convention $t_0 =t.$ Note that,
\begin{equation*}
\begin{split}
\| \tilde{\Xi}_{m,t}^{(k)}\|_{\mathrm{H}^{\alpha}_k} \leq & \int_{0}^{t} \mathrm{d} t_1 \int_0^{t_1} \mathrm{d} t_2 \cdots \int_0^{t_{m-1}} \mathrm{d} t_m \Big \| e^{\mathrm{i} (t - t_1) \triangle^{(k)}_{\pm}} \tilde{B}^{(k)} \cdots \\
& \times e^{\mathrm{i} (t_{m-1} - t_m) \triangle^{(k+m-1)}_{\pm}} \tilde{B}^{(k+m-1)} \gamma^{(k+m)}_{t_m} \Big \|_{\mathrm{H}^{\alpha}_k}\\
\leq & \int_{0}^{t}\mathrm{d}t_1\int_0^{t_1}\mathrm{d}t_2\cdots\int_0^{t_{m-1}}\mathrm{d} t_m~ k \cdots \\
& \quad \times (k+m-1) (C_{n, \alpha})^{m} \| \gamma^{(k+m)}_{t_m} \|_{\mathrm{H}^{\alpha}_{k+m}}\\
\leq &  m (m+1) \cdots (2m-1) (C_{n, \alpha})^{m} \\
& \quad \times \int_{0}^t \mathrm{d}t_1\int_0^{t_1}\mathrm{d}t_2\cdots\int_0^{t_{m-1}} \| \gamma^{(k+m)}_{t_m} \|_{\mathrm{H}^{\alpha}_{k+m}} \mathrm{d} t_m .\\
\end{split}
\end{equation*}
Set $T = 1 / [4 C_{n, \alpha} \; \| \Gamma (t) \|_{C ( [0,T_0]; \mathcal{H}^{\alpha})} ].$ Then, combining the above estimate and \eqref{eq:XiEsitmate} yields
\be
\begin{split}
\|\gamma^{(k)}_t \|_{C([0, T]; \mathrm{H}^{\alpha}_k) }
& \leq \sum_{j = 0}^{m-1} \binom{k+j-1}{j} (C_{n, \alpha}  T )^j\|\gamma^{(k+j)}_{0}\|_{\mathrm{H}^{\alpha}_{k+j}}\\
& \; + m (m+1) \cdots (2m-1) (C_{n, \alpha})^{m} \\
& \quad \times \int_{0}^T \mathrm{d}t_1\int_0^{t_1}\mathrm{d}t_2\cdots\int_0^{t_{m-1}} \| \gamma^{(k+m)}_{t_m} \|_{\mathrm{H}^{\alpha}_{k+m}} \mathrm{d} t_m \\
& \le \sum_{j = 0}^{\8} \binom{k+j-1}{j} \Big ( \frac{1}{4 \| \Gamma_0 \|_{\mathcal{H}^{\alpha}}} \Big)^j\|\gamma^{(k+j)}_{0}\|_{\mathrm{H}^{\alpha}_{k+j}}\\
& \quad + \big ( C_{n, \alpha} T \big )^m \binom{2m-1}{m} \| \gamma^{(k+m)}_t \|_{C ([0,T]; \mathrm{H}^{\alpha}_{k+m})}\\
& \le \sum_{j = 0}^{\8} \binom{k+j-1}{j} \Big ( \frac{1}{4 \| \Gamma_0 \|_{\mathcal{H}^{\alpha}}} \Big)^j\|\gamma^{(k+j)}_{0}\|_{\mathrm{H}^{\alpha}_{k+j}}\\
& \quad + \frac{\| \Gamma (t) \|^k_{C ( [0,T_0]; \mathcal{H}^{\alpha})}}{4^m} \binom{2m-1}{m} \frac{\| \gamma^{(k+m)}_t \|_{C ([0,T]; \mathrm{H}^{\alpha}_{k+m})}}{\| \Gamma (t) \|^{k+m}_{C ( [0,T_0]; \mathcal{H}^{\alpha})}}.
\end{split}
\ee
Since $\binom{2m-1}{m} \approx \frac{4^m}{\sqrt{m}},$ taking $m \to \8$ we conclude that
\be
\|\gamma^{(k)}_t \|_{C([0, T]; \mathrm{H}^{\alpha}_k) } \le \sum_{j = 0}^{\8} \binom{k+j-1}{j} \Big ( \frac{1}{4 \| \Gamma_0 \|_{\mathcal{H}^{\alpha}}} \Big)^j\|\gamma^{(k+j)}_{0}\|_{\mathrm{H}^{\alpha}_{k+j}}.
\ee
Then, taking $\lambda = 4 \| \Gamma_0 \|_{\mathcal{H}^{\alpha}}$ we have
\begin{equation}\label{eq:GammaSum_m-itera}
\begin{split}
\sum_{k=1}^{\8} \frac{1}{\lambda^k} \|\gamma^{(k)}_t \|_{C ([0,T]; \mathrm{H}^{\alpha}_k)}
& \leq  \sum_{j = 0}^{\8} \sum_{k=1}^{\8} \binom{k+j-1}{j} (C_{n, \alpha} T \lambda )^j  \frac{1}{\lambda^{k+j}}  \|\gamma^{(k+j)}_{0}\|_{\mathrm{H}^{\alpha}_{k+j}}\\
& \le \sum_{\el \ge 1} \Big ( \frac{1}{2\| \Gamma_0 \|_{\mathcal{H}^{\alpha}} } \Big )^{\el} \|\gamma_0^{(\el)} \|_ {\mathcal{H}_{\el}^{\alpha}} \le 1.
\end{split}
\end{equation}
This concludes that $\| \Gamma (t) \|_{C ( [0,T]; \mathcal{H}^{\alpha})} \le 2 \| \Gamma_0 \|_{\mathcal{H}^{\alpha}}.$ This proves \eqref{eq:SpacetimeEstimate-Stability}.
\end{proof}

\begin{remark}\label{rk:PfThLocalWellposued1}\rm
From the above proof we find that the assumption that functions $\gamma$ are symmetric can be dropped in Theorem \ref{th:LocalWellposued-alpha>n/2}.
\end{remark}

\section{Proof of Theorem \ref{th:LocalWellposued-alpha>(n-1)/2}}\label{PfLocalWellposued2}

For simplicity, we denote by $L^2_t (\mathrm{H}^{\alpha}_k) = L^2 (\mathbb{R}, \mathrm{H}^{\alpha}_k),$ equipped with the norm
\begin{equation*}
\| f \|_{L^2_t (\mathrm{H}^{\alpha}_k)} = \Big ( \int_{\mathbb{R}} \| f (t ) \|^2_{\mathrm{H}^{\alpha}_k} d t \Big )^{\frac{1}{2}}.
\end{equation*}
Evidently, $L^2_t (\mathrm{H}^{\alpha}_k)$ is a Hilbert space.

\begin{lemma}\label{le:BOperatorEstimate}{\rm (cf. \cite[Proposition A.1]{CP2010})}
Assume that  $n \geq 2$ and $\alpha > (n-1) / 2.$ Then there exists a constant $ C_{n, \alpha}$ depending only on $n$ and $\alpha$ such
that, for any symmetric $\gamma^{(k+1)} \in \mathcal{S} (\mathbb{R}^{(k+1) n} \times \mathbb{R}^{(k+1) n}),$
\begin{equation}\label{eq:BoperatorEstimate}
\| B_{j,k} e^{\mathrm{i} t \triangle^{(k+1)}_{\pm}} \gamma^{(k+1)} \|_{L^2_t(\mathrm{H}^{\alpha}_k)} \leq C_{n, \alpha} \|\gamma^{(k+1)} \|_{\mathrm{H}^{\alpha}_{k+1}},
\end{equation}
for all $k \geq 1,$ where $j=1,\cdots,k.$

Consequently, $B^{(k)}$ can be extended to a bounded operator from $\mathrm{H}^{\alpha}_{k+1}$ to $\mathrm{H}^{\alpha}_k,$ still denoted by $B^{(k)},$ satisfying
\be
\| B^{(k)} e^{\mathrm{i} t \triangle^{(k+1)}_{\pm}} \gamma^{(k+1)} \|_{L^2_t(\mathrm{H}^{\alpha}_k)} \leq C_{n, \alpha} k \|\gamma^{(k+1)} \|_{\mathrm{H}^{\alpha}_{k+1}},
\ee
for all $\gamma^{(k+1)} \in \mathrm{H}^{\alpha}_{k+1}.$
\end{lemma}

As in \cite{KM2008}, we introduce the notation
\be
D^{(k)}_j ( \Gamma) (t) : = \int^t_0 \cdots \int^{t_{j-1}}_0 J^{(k)}_j ( {\bf t}_j) \gamma^{(k+j+1)} (t_j) d t_1 \cdots d t_j
\ee
for $k, j \ge 1,$ where
\beq\label{eq:DuhamelExpansion}
\begin{split}
J^{(k)}_j ( {\bf t}_j) = \prod^j_{i=1} e^{\mathrm{i} (t_{i-1}-t_i) \triangle^{(k+i)}_{\pm}} \tilde{B}^{(k+i)}
\end{split}
\eeq
with ${\bf t}_j = (t, t_1, \ldots , t_j)$ and with the convention $t_0 =t.$

The following lemma is crucial for the proof of Theorem \ref{th:LocalWellposued-alpha>(n-1)/2}.

\begin{lemma}\label{le:DuhamelEstimate}
Assume that $n \geq 2$ and $\alpha > (n-1)/ 2.$ Then there exist an absolute constant $A>2$ and a constant $C_{n,\alpha}>0$ depending only on $n$ and $\alpha$ so that the estimates below hold
\begin{enumerate}[{\rm (1)}]

\item For any $\Gamma_0 = \{ \gamma^{(k)}_0 \}_{k \ge 1} \in \bigotimes^{\8}_{k=1} \mathrm{H}^{\alpha}_k,$
\begin{equation}\label{eq:DuhamelEstimate0}
\begin{split}
\Big \| \tilde{B}^{(k)} D^{(k)}_j ( e^{\mathrm{i} t \triangle_{\pm}} \Gamma_0 ) (t) \Big \|_{L^1_{t \in [0,T]} \mathrm{H}^{\alpha}_k}
\leq k A^{k+j} (C_{n, \alpha} T)^{\frac{j + 1}{2}} \|\gamma^{(k+j +1)}_0 \|_{\mathrm{H}^{\alpha}_{k+j+1}},
\end{split}
\end{equation}
for $k, j \ge 1$ and $T>0.$

\item For any $T>0$ and $\Gamma (t) = \{\gamma^{(k)}_t \}_{k \geq 1}$ with $\gamma^{(k)}_t \in L^1_{t \in[0,T]} \mathrm{H}^{\alpha}_k,$
\begin{equation}\label{eq:DuhamelEstimateT}
\begin{split}
\Big \| \tilde{B}^{(k)} & D^{(k)}_m ( \Gamma) (t) \Big \|_{L^1_{t \in [0,T]} \mathrm{H}^{\alpha}_k}
\leq k A^{k+m}(C_{n, \alpha} T)^{\frac{m}{2}} \| B^{(k+m)} \gamma^{(k+m +1)} (t) \|_{L^1_{t \in [0,T]}\mathrm{H}^{\alpha}_{k+m}},
\end{split}
\end{equation}
for $k, m \ge 1.$

\end{enumerate}
\end{lemma}

\begin{proof}
The inequalities \eqref{eq:DuhamelEstimate0} and \eqref{eq:DuhamelEstimateT} can be proved by using the so-called ``board game" argument presented in \cite{KM2008}. For the details see the proof of Proposition A.2 in \cite{CP2010}.
\end{proof}

Now we are ready to prove Theorem \ref{th:LocalWellposued-alpha>(n-1)/2}. To this end, we introduce the system
\beq\label{eq:GPHierarchyFunctIntEquaXi}
\Gamma (t) = e^{\mathrm{i} t \hat{\triangle}_{\pm}} \Gamma_0 + \int^{t}_{0} d s~  e^{\mathrm{i} (t-s) \hat{\triangle}_{\pm}} \Xi_s\;,
\eeq
\beq\label{eq:GPStrongSolutionFunctXi}\begin{split}
\Xi_t = \hat{B} e^{\mathrm{i} t \hat{\triangle}_{\pm}} \Gamma_0 + \int^{t}_{0} d s\, \hat{B} e^{\mathrm{i} (t-s) \hat{\triangle}_{\pm}} \Xi_s\; ,
\end{split}\eeq
which is formally equivalent to the system \eqref{eq:GPHierarchyFunctIntEquaGamma}, \eqref{eq:GPStrongSolutionFunctGamma}. The proof is divided into three parts as follows. As before, it suffices to consider the case $t \ge 0.$

\begin{proof}
(1)\; Let $\alpha > (n-1)/2$ and $n \ge 2.$ Let $\Gamma_0 = \{ \gamma^{(k)}_0 \}_{k \ge 1} \in \mathcal{H}^{\alpha}$ and $\Xi_0 = \{ \rho^{(k)}_0 (t) \}_{k \ge 1} =0.$ Given $k \ge 1,$ for any $m \ge 1$ we define
\beq\label{eq:m-iteraDf}
\rho^{(k)}_m (t) = \tilde{B}^{(k)} e^{\mathrm{i} t \triangle^{(k+1)}_{\pm}} \gamma^{(k+1)}_0 + \int^t_0 d s \tilde{B}^{(k)}  e^{\mathrm{i} (t-s) \triangle^{(k+1)}_{\pm}} \rho^{(k+1)}_{m-1} (s)
\eeq
for $t \in [0, T],$ where $T$ will be fixed later. Set $\Xi_m (t) = \{ \rho^{(k)}_m (t) \}_{k \ge 1}$ for every $m \ge 1.$ By expansion, for every $m \ge 2$ one has
\beq\label{eq:DuhamelExpanFunct_m-itera}
\begin{split}
\rho^{(k)}_m (t) = & \tilde{B}^{(k)} e^{\mathrm{i} t \triangle^{(k+1)}_{\pm}} \gamma^{(k+1)}_0\\
&\; + \sum^{m-1}_{j=1}  \tilde{B}^{(k)} \int^t_0 \cdots \int^{t_{j-1}}_0 d t_1 \cdots d t_j e^{\mathrm{i} (t -t_1) \triangle^{(k+1)}_{\pm}} \tilde{B}^{(k+1)}\\
& \quad \times \cdots e^{\mathrm{i} (t_{j-1}-t_j) \triangle^{(k+j)}_{\pm}} \tilde{B}^{(k+j)} e^{\mathrm{i} t_j \triangle^{(k+j+1)}_{\pm}} \gamma^{(k+j+1)}_0\\
\end{split}
\eeq
with the convention $t_0 =t,$ that is,
\be\label{eq:DuhamelExpanFunct_m-itera}
\begin{split}
\rho^{(k)}_m (t) = \sum^{m-1}_{j=0} \tilde{B}^{(k)} D^{(k)}_j (\Gamma_0) (t)
\end{split}
\ee
with the convenience $D^{(k)}_0 (\Gamma_0) (t) = e^{\mathrm{i} t \triangle^{(k+1)}_{\pm}} \gamma^{(k+1)}_0.$ By Lemma \ref{le:DuhamelEstimate} (1) we have
\be\begin{split}
\Big \| \sum^{m-1}_{j=0} \tilde{B}^{(k)} D^{(k)}_j (\Gamma_0) (t)\Big \|_{L^1_{t \in [0,T]} \mathrm{H}^{\alpha}_k}
\le k \sum^{m-1}_{j=0} A^{k+j} ( \sqrt{C_{n, \alpha} T})^{j+1} \big \| \gamma^{(k+j+1)}_0 \big \|_{\mathrm{H}^{\alpha}_{k+j+1}}.\\
\end{split}\ee
Then,
\beq\label{eq:m-iteraEstimate}\begin{split}
\| \rho^{(k)}_m \|_{L^1_{t \in [0,T]} \mathrm{H}^{\alpha}_k} \le  k \sum^{m-1}_{j=0} A^{k+j} ( \sqrt{C_{n, \alpha} T})^{j+1} \big \| \gamma^{(k+j+1)}_0 \big \|_{\mathrm{H}^{\alpha}_{k+j+1}}.\\
\end{split}\eeq

Set $T :=1 / [C_{n, \alpha} A^2 \| \Gamma_0 \|^2_{\mathcal{H}^{\alpha}}].$ For $\lambda>0$ one has by \eqref{eq:m-iteraEstimate}
\be\begin{split}
\sum_{k \ge 1} \frac{1}{\lambda^k}  \| \rho^{(k)}_m \|_{L^1_{t \in [0,T]} \mathrm{H}^{\alpha}_k}
\le  \sum_{k \ge 1} \Big ( \frac{2 A}{\lambda} \Big )^k \sum^{m-1}_{j=0} \frac{1}{\| \Gamma_0 \|_{\mathcal{H}^{\alpha}}^{j+1}} \big \| \gamma^{(k+j+1)}_0 \big \|_{\mathrm{H}^{\alpha}_{k+j+1}}.
\end{split}
\ee
Choosing $\lambda = 4 A \| \Gamma_0 \|_{\mathcal{H}^{\alpha}},$ we have
\be\begin{split}
\sum_{k \ge 1} \frac{1}{\lambda^k} \| \rho^{(k)}_m \|_{L^1_{t \in [0,T]} \mathrm{H}^{\alpha}_k}
\le \sum_{k \ge 1} \frac{1}{2^k} \sum_{j \ge 1} \frac{1}{\| \Gamma_0 \|_{\mathcal{H}^{\alpha}}^{k+j+1}} \big \| \gamma^{(k+j+1)}_0 \big \|_{\mathrm{H}^{\alpha}_{k+j+1}} \le 1.
\end{split}\ee
This concludes that for every $m \ge 1,$ $\Xi_m \in L^1_{t \in [0,T]} \mathcal{H}^{\alpha}$ and
\beq\label{eq:m-iteraApriorBound}
\| \Xi_m \|_{L^1_{t \in [0,T]} \mathcal{H}^{\alpha}} \le 4 A \| \Gamma_0 \|_{\mathcal{H}^{\alpha}}.
\eeq

Now, for fixed $k \ge 1$ and any $n, m$ with $n > m$ we have
\be\begin{split}
\| \rho^{(k)}_m  - \rho^{(k)}_n \|_{L^1_{t \in [0,T]} \mathrm{H}^{\alpha}_k}
\le & \sum^{n-1}_{j=m} (2A)^{k+j} ( \sqrt{C_{n, \alpha} T})^{j+1} \big \| \gamma^{(k+j+1)}_0 \big \|_{\mathrm{H}^{\alpha}_{k+j+1}}\\
\le & (2 A \| \Gamma_0 \|_{\mathcal{H}^{\alpha}})^k \sum_{j \ge m} \frac{1}{\| \Gamma_0 \|_{\mathcal{H}^{\alpha}}^{k+j+1}} \big \| \gamma^{(k+j+1)}_0 \big \|_{\mathrm{H}^{\alpha}_{k+j+1}}.
\end{split}\ee
This concludes that for each $k \ge 1,$ $\rho^{(k)}_m$ converges in $L^1_{t \in [0,T]} \mathrm{H}^{\alpha}_k$ as $m \to \8,$
whose limitation is denoted by $\rho^{(k)}.$

Set $\Xi (t) = \{\rho^{(k)} (t) \}_{k \ge 1}.$ Note that for any $n,m \ge 1,$
\be\begin{split}
\Big \| \int^t_0 d s & \tilde{B}^{(k)}  e^{\mathrm{i} (t-s) \triangle^{(k+1)}_{\pm}} [ \rho^{(k+1)}_{m-1} (s) - \rho^{(k+1)}_{n-1} (s) ]\Big \|_{L^1_{t \in [0,T]} \mathrm{H}^{\alpha}_k}\\
\le & \sum^k_{\el =1} \int^T_0 \int^T_0 d t d s \big \| B_{\el, k} e^{\mathrm{i} (t-s) \triangle^{(k+1)}_{\pm}} [ \rho^{(k+1)}_{m-1} (s) - \rho^{(k+1)}_{n-1} (s) ] \big \|_{\mathrm{H}^{\alpha}_k}\\
\le & T^{1/2} \sum^k_{\el =1} \int^T_0 d s \big \| B_{\el, k} e^{\mathrm{i} (t-s) \triangle^{(k+1)}_{\pm}} [ \rho^{(k+1)}_{m-1} (s) - \rho^{(k+1)}_{n-1} (s) ] \big \|_{L^2_{t \in [0,T]}\mathrm{H}^{\alpha}_k}\\
\le & C_{n, \alpha} k T^{1/2} \big \| \rho^{(k+1)}_{m-1} - \rho^{(k+1)}_{n-1} \big \|_{L^1_{t \in [0,T]} \mathrm{H}^{\alpha}_{k+1}},
\end{split}
\ee
where we have used the Cauchy-Schwarz inequality with respect to the integral in $t$ in the second inequality and used Lemma \ref{le:BOperatorEstimate} in the last inequality. Thus, taking $m \to \8$ in \eqref{eq:m-iteraDf} we prove that $\Xi$ is a solution to \eqref{eq:GPStrongSolutionFunctXi}.
Moreover, taking $m \to \8$ in \eqref{eq:m-iteraApriorBound} we obtain \eqref{eq:SpacetimeEstimate}.

(2)\; Fix $T_0 >0.$ Suppose $\Gamma (t) \in C([0,T_0], \mathcal{H}^{\alpha} )$ is a solution to \eqref{eq:GPHierarchyEquaFunct} so that $\hat{B} \Gamma (t) \in L^1_{t \in [0,T_0]}\mathcal{H}^{\alpha}.$ Given $ T \in (0, T_0],$ which will be fixed later. By \eqref{eq:GPStrongSolutionDuh-kth}, for every $m \geq 1$ one has
\be\label{eq:DuhamelExpanFunct}
\begin{split}
\tilde{B}^{(k)}\gamma^{(k+1)}_t = & \tilde{B}^{(k)} e^{\mathrm{i} t \triangle^{(k+1)}_{\pm}} \gamma^{(k+1)}_0\\
&\; + \sum^{m-1}_{j=1}  \tilde{B}^{(k)} \int^t_0 \cdots \int^{t_{j-1}}_0 d t_1 \cdots d t_j e^{\mathrm{i} (t -t_1) \triangle^{(k+1)}_{\pm}} \tilde{B}^{(k+1)}\\
& \quad \times \cdots e^{\mathrm{i} (t_{j-1}-t_j) \triangle^{(k+j)}_{\pm}} \tilde{B}^{(k+j)} e^{\mathrm{i} t_j \triangle^{(k+j+1)}_{\pm}} \gamma^{(k+j+1)}_0\\
& \; + \tilde{B}^{(k)} \int^t_0 \int^{t_1}_0 \cdots \int^{t_{m-1}}_0 d t_1 d t_2 \cdots d t_m e^{\mathrm{i} (t -t_1) \triangle^{(k+1)}_{\pm}} \tilde{B}^{(k+1)}\\
& \quad \times \cdots e^{\mathrm{i} (t_{m-1}-t_m) \triangle^{(k+m)}_{\pm}} \tilde{B}^{(k+m)} \gamma^{(k+m+1)}(t_m),
\end{split}
\ee
with the convention $t_0 =t,$ that is,
\be\label{eq:DuhamelExpanFunct}
\begin{split}
\tilde{B}^{(k)}\gamma^{(k+1)}_t =  \tilde{B}^{(k)} e^{\mathrm{i} t \triangle^{(k+1)}_{\pm}} \gamma^{(k+1)}_0
 + \sum^{m-1}_{j=1} \tilde{B}^{(k)} D^{(k)}_j (\Gamma_0) (t) + \tilde{B}^{(k)} D^{(k)}_m (\Gamma) (t).
\end{split}
\ee
By Lemma \ref{le:DuhamelEstimate} we have
\begin{equation*}
\begin{split}
\big \| \tilde{B}^{(k)} \gamma^{(k+1)}_t \big \|_{L^1_{t \in [0,T]} \mathrm{H}^{\alpha}_k}
& \le \sum^{m-1}_{j=0} k A^{k+j} ( \sqrt{C_{n, \alpha} T})^{j+1} \big \| \gamma^{(k+j+1)}_0 \big \|_{\mathrm{H}^{\alpha}_{k+j+1}}\\
&\quad + k A^{k+m} (\sqrt{C_{n, \alpha} T})^m \big \| B^{(k+m)} \gamma^{(k+m+1)}_t \big \|_{L^1_{t \in [0,T]}\mathrm{H}^{\alpha}_{k+m}}.
\end{split}
\end{equation*}
Set $T = 1 / [C_{n, \alpha} A^2 \max \{ \| \hat{B} \Gamma (t) \|^2_{L^1_{t \in [0,T_0]}\mathcal{H}^{\alpha}}, \| \Gamma_0 \|^2_{\mathcal{H}^{\alpha}} \}].$ Taking $m \to \8$ we have
\be
\big \| \tilde{B}^{(k)} \gamma^{(k+1)}_t \big \|_{L^1_{t \in [0,T]} \mathrm{H}^{\alpha}_k} \le (2A)^k \sum^{\8}_{j=0} \frac{1}{\| \Gamma_0 \|_{\mathcal{H}^{\alpha}}^{j+1}} \big \| \gamma^{(k+j+1)}_0 \big \|_{\mathrm{H}^{\alpha}_{k+j+1}}
\ee
Then, for $\lambda>0$ we have
\be
\begin{split}
\sum_{k \ge 1} \frac{1}{\lambda^k} & \big \|\tilde{B}^{(k)} \gamma^{(k+1)}_t \big \|_{L^1_{t \in [0,T]} \mathrm{H}^{\alpha}_k}
\le \sum_{k \ge 1} \sum^{\8}_{j=0} \Big ( \frac{2A}{\lambda}\Big )^k \frac{1}{\| \Gamma_0 \|_{\mathcal{H}^{\alpha}}^{j+1}} \big \| \gamma^{(k+j+1)}_0 \big \|_{\mathrm{H}^{\alpha}_{k+j+1}}.
\end{split}
\ee
Choose $\lambda = 4 A \|\Gamma_0\|_{\mathcal{H}^{\alpha}}.$ Then we have
\be\begin{split}
\sum_{k \ge 1} \frac{1}{\lambda^k} \big \| \tilde{B}^{(k)} \gamma^{(k+1)}_t \big \|_{L^1_{t \in [0,T]} \mathrm{H}^{\alpha}_k}
\le \sum_{k \ge 1} \frac{1}{2^k} \sum^{\8}_{j=0} \frac{1}{\| \Gamma_0 \|_{\mathcal{H}^{\alpha}}^{k+j+1}} \big \| \gamma^{(k+j+1)}_0 \big \|_{\mathrm{H}^{\alpha}_{k+j+1}} \le 1.
\end{split}\ee
This completes the proof of (2).

(3)\; Fix $T_0 >0.$ Suppose $\Gamma (t), \Gamma'(t) \in C ([0,T_0], \mathcal{H}^{\alpha})$ are two solutions to \eqref{eq:GPHierarchyEquaFunct} such that $\hat{B} \Gamma (t)$ and $\hat{B} \Gamma' (t)$ both belong to $L^1_{t \in [0,T_0]}\mathcal{H}^{\alpha}.$ Since \eqref{eq:GPHierarchyEquaFunct} is linear, it suffices to consider $\Gamma (t)$ instead of $\Gamma (t) - \Gamma'(t).$ Choosing
\be
T = 1 / [C_{n, \alpha} A^2 \max \{ \| \hat{B} \Gamma (t) \|^2_{L^1_{t \in [0,T_0]}\mathcal{H}^{\alpha}}, \| \Gamma_0 \|^2_{\mathcal{H}^{\alpha}} \}],
\ee
by (2) we have
\be
\| \Gamma (t) \|_{C ([0,T], \mathcal{H}^{\alpha})} \le \|\Gamma_0 \|_{\mathcal{H}^{\alpha}} +  4 A \|\Gamma_0 \|_{\mathcal{H}^{\alpha}} = (1+ 4 A) \|\Gamma_0 \|_{\mathcal{H}^{\alpha}}.
\ee
This completes the proof.
\end{proof}

\begin{remark}\label{rk:n=3}\rm
For $n=3,$ using \cite[Theorem 1.3]{KM2008} instead of Lemma \ref{le:BOperatorEstimate} we can prove as done above that the Cauchy problem \eqref{eq:GPHierarchyEquaFunct} is locally well posed in $\mathcal{H}^1$ in the sense of Theorem \ref{th:LocalWellposued-alpha>(n-1)/2}. We omit the details.
\end{remark}

\section{The quintic Gross--Pitaevskii hierarchy}\label{GPHierarchyQuintic}

In this section, we consider the so-called quintic Gross-Pitaevskii hierarchy. Recall that the quintic Gross-Pitaevskii hierarchy $\Gamma(t)= ( \gamma^{(k)}(t) )_{k\geq 1}$ is given by
\beq\label{eq:GPHierarchyEquaQuintic}
\mathrm{i} \partial_t \gamma^{(k)}_t = \Big [\sum^k_{j=1} (- \Delta_{x_j}), \gamma^{(k)}_t \Big ] + \mu Q^{(k)} \gamma^{(k+2)}_t,\quad \mu = \pm 1,
\eeq
in $n$ dimensions, for $k \in \mathbb{N},$ where the operator $Q^{(k)}$ is defined by
\be
Q^{(k)} \gamma^{(k+2)}_t = \sum^k_{j=1} \mathrm{T r}_{k+1, k+2} \left [ \delta (x_j - x_{k+1}) \delta (x_j - x_{k+2}), \gamma^{(k+2)}_t \right ].
\ee
It is {\it defocusing} if $\mu =1,$ and {\it focusing} if $\mu= -1.$ We note that the quintic GP hierarchy accounts for $3$-body interactions between the Bose particles (see \cite{CP2011} and references therein for details).

In terms of kernel functions, the Cauchy problem for the quintic GP hierarchy \eqref{eq:GPHierarchyEquaQuintic} can be written as follows
\beq\label{eq:GPHierarchyEquaFunctQuintic}
\left \{ \begin{split}
& \big ( \mathrm{i} \partial_t + \triangle^{(k)} \big ) \gamma^{(k)}_t ({\bf x}_k;{\bf x}'_k) = \mu \big ( Q^{(k)} \gamma^{(k+2)}_t \big ) ({\bf x}_k; {\bf x}'_k ),\\
& \gamma^{(k)}_{t=0} ({\bf x}_k;{\bf x}'_k) = \gamma^{(k)}_0 ({\bf x}_k;{\bf x}'_k),\; k \in \mathbb{N},
\end{split}\right.
\eeq
where $Q^{(k)}: = \sum^k_{j=1} Q^{(k)}_j$ with the action of $Q^{(k)}_j = Q^{(k)}_{j, +} - Q^{(k)}_{j, -}$ on $\gamma^{(k+2)} ({\bf x}_{k+2}, {\bf x}'_{k+2}) \in \mathcal{S} (\mathbb{R}^{(k+2)n} \times \mathbb{R}^{(k+2)n})$ being defined according to
\be\begin{split}
 \big ( Q^{(k)}_{j, +} & \gamma^{(k+2)} \big ) ({\bf x}_k, {\bf x}'_k)\\
& : = \int d x_{k+1} d x_{k+2} d x'_{k+1} d x'_{k + 2} \gamma^{(k+2)} ({\bf x}_k, x_{k+1}, x_{k+2}; {\bf x}'_k, x'_{k+1}, x'_{k+2}) \\
& \; \quad \times \delta (x_j - x_{k+1}) \delta (x_j - x'_{k+1}) \delta (x_j - x_{k+2}) \delta (x_j - x'_{k+2})\\
& = \gamma^{(k+2)} ({\bf x}_k, x_j, x_j; {\bf x}'_k, x_j, x_j ),
\end{split}\ee
and
\be\begin{split}
 \big ( Q^{(k)}_{j, -} & \gamma^{(k+2)} \big ) ({\bf x}_k, {\bf x}'_k)\\
& : = \int d x_{k+1} d x_{k+2} d x'_{k+1} d x'_{k + 2} \gamma^{(k+2)} ({\bf x}_k, x_{k+1}, x_{k+2}; {\bf x}'_k, x'_{k+1}, x'_{k+2}) \\
& \; \quad \times \delta (x'_j - x_{k+1}) \delta (x'_j - x'_{k+1}) \delta (x'_j - x_{k+2}) \delta (x'_j - x'_{k+2})\\
& = \gamma^{(k+2)} ({\bf x}_k, x'_j, x'_j; {\bf x}'_k, x'_j, x'_j ),
\end{split}\ee
for $j=1, \ldots, k.$

Let $\varphi \in \mathrm{H}^1(\mathbb{R}^n),$ then one can easily verify that a particular solution to \eqref{eq:GPHierarchyEquaFunctQuintic} with initial conditions
\be
\gamma^{(k)}_{t=0}({\bf x}_k; {\bf x}'_k) = \prod^k_{j=1} \varphi(x_j) \overline{\varphi(x'_j)},\quad k=1,2,\ldots,
\ee
is given by
\be
\gamma^{(k)}_t ({\bf x}_{k}; {\bf x}'_{k} ) = \prod^k_{j=1} \varphi_t (x_j) \overline{\varphi_t ( x'_j )},\quad k=1,2,\ldots,
\ee
where $\varphi_t$ satisfies the quintic non-linear Schr\"odinger equation
\beq\label{eq:GPEquaQuintic}
\mathrm{i} \partial_t \varphi_t = -\Delta \varphi_t + \mu |\varphi_t|^4 \varphi_t,\quad \varphi_{t=0}=\varphi.
\eeq

The GP hierarchy \eqref{eq:GPHierarchyEquaFunctQuintic} can be written in the integral form
\beq\label{eq:GPHierarchyIntEquaQuintic}
\gamma^{(k)}_t = e^{ \mathrm{i} t{\Delta}^{(k)}}  \gamma^{(k)}_0 + \int^{t}_{0} d s~ e^{ \mathrm{i} (t-s){\Delta}^{(k)}}  \tilde{Q}^{(k)} \gamma^{(k+2)}_s,\; k=1,2,\ldots,
\eeq
where $\tilde{Q}^{(k)} = - \mathrm{i} \mu Q^{(k)}.$ Evidently, such a solution can be obtained by solving the following infinity linear hierarchy of integral equations
\begin{equation}\label{eq:GPStrongSolutionFunctQuintic}
\tilde Q^{(k)}\gamma^{(k+2)}_t = \tilde Q^{(k)}e^{ \mathrm{i} t{\Delta}^{(k+2)}} \gamma^{(k+2)}_0 + \int^{t}_{0} d s~ \tilde{Q}^{(k)}e^{ \mathrm{i} (t-s){\Delta}^{(k+2)}} \tilde{Q}^{(k+2)}\gamma^{(k+4)}_s,
\end{equation}
for any $k\geq 1.$ If we write
\be
\hat \Delta \Gamma:= ( \Delta^{(k)}\gamma^{(k)} )_{k\geq 1} \quad \text{and} \quad \hat{Q} \Gamma: = ( \tilde{Q}^{(k)}\gamma^{(k+2)} )_{k\geq 1} ,
\ee
then \eqref{eq:GPHierarchyIntEquaQuintic} and \eqref{eq:GPStrongSolutionFunctQuintic} can be written as
\begin{equation}\label{Gamma:1}
\Gamma(t)=e^{ \mathrm{i} t\hat{\Delta}}\Gamma_0+\int_0^tds~ e^{ \mathrm{i} (t-s)\hat{\Delta}} \hat{Q} \Gamma(s),
\end{equation}
and
\begin{equation}\label{Gamma:2}
\hat Q\Gamma(t) =\hat Qe^{ \mathrm{i} t\hat{\Delta}}\Gamma_0+\int_0^tds~ \hat Q e^{ \mathrm{i} (t-s)\hat{\Delta}}\hat Q\Gamma(s),
\end{equation}
respectively.

Formally we can expand the solution $\gamma^{(k)}_t$ of \eqref{eq:GPHierarchyIntEquaQuintic}
for any $m \geq 2$ as
\beq\label{eq:DuhamelExpanQuintic}
\begin{split}
\gamma^{(k)}_t = & e^{ \mathrm{i} t{\Delta}^{(k)}}  \gamma^{(k)}_0 + \sum^{m-1}_{j=1} \int^t_0 d t_2 \int^{t_2}_0 d t_4 \cdots
\int^{t_{2(j-1)}}_0 d t_{2 j}  e^{ \mathrm{i} (t-t_2 ){\Delta}^{(k)}} \tilde{Q}^{(k)} \cdots\\
& \; \times e^{ \mathrm{i} ( t_{2 (j-1)} - t_{2 j} ){\Delta}^{(k+2(j-1))}} \tilde{Q}^{(k+2(j-1))}   e^{ \mathrm{i} t_j {\Delta}^{(k+2j)}}\gamma^{(k+2j)}_0\\
& \; + \int^t_0 d t_2 \int^{t_2}_0 d t_4 \cdots \int^{t_{2(m-1)}}_0 d t_{2 m} e^{ \mathrm{i} (t-t_2){\Delta}^{(k)}} \tilde{Q}^{(k)} \cdots \\
& \; \times e^{ \mathrm{i} ( t_{2(m-1)} - t_{2 m} ){\Delta}^{(k+2(m-1))}} \tilde{Q}^{(k+2(m-1))} \gamma^{(k+2m)}_{t_m},
\end{split}
\eeq
with the convention $t_0 =t.$

Let us make the notion of solution more precise for the quintic GP hierarchy \eqref{eq:GPHierarchyEquaFunctQuintic}.

\begin{definition}\label{df:StrongSolutionQuintic}
A function $\Gamma (t) = ( \gamma^{(k)}_t )_{k \geq 1}: I \mapsto \mathcal{H}^{s}$ on a non-empty time interval $0 \in I \subset \mathbb{R}$ is said to be a local $(strong)$ solution
to the Gross-Pitaevskii hierarchy \eqref{eq:GPHierarchyEquaFunctQuintic} if it lies in the class $C (K, \mathcal{H}^{s})$ for all compact sets $K \subset I$ and obeys the Duhamel formula
\beq\label{eq:MildSolutionQuintic}
\gamma^{(k)}_t = e^{\mathrm{i} t \triangle^{(k)}} \gamma^{(k)}_0 - \mathrm{i} \mu \int^{t}_{0} d s\; e^{\mathrm{i} (t-s) \triangle^{(k)}} Q^{(k)} \gamma^{(k+2)}_s,\quad \forall t \in I,
\eeq
holds in $\mathrm{H}^{s}_k$ for every $k=1,2,\ldots.$
\end{definition}

We will prove that the Cauchy problem \eqref{eq:GPHierarchyEquaFunctQuintic} is locally well-posed for $s > \max \{\frac{1}{2}, (n-1)/2 \}.$

\begin{theorem}\label{th:LocalWellposuedQuintic1}
Assume that $n \geq 1$ and $s > \frac{n}{2}.$  The Cauchy problem  \eqref{eq:GPHierarchyEquaFunctQuintic} is locally well posed. More precisely, there exists a constant $ K_{ n,s} > 0$ depending only on $n$ and $s$ such that
\begin{enumerate}[{\rm (1)}]

\item For every $\Gamma_0= ( \gamma_0^{(k)} )_{k\geq 1}\in \mathcal{H}^{s},$ let $I = [-T, T]$ with $T= \frac{K_{n,s}}{\|\Gamma_0\|_{\mathcal{H}^{s}}^2}.$ Then there exists a solution $\Gamma(t)= ( \gamma_t^{(k)} )_{k\geq 1}\in C(I,\mathcal{H}^{s})$ to the Gross-Pitaevskii hierarchy \eqref{eq:GPHierarchyEquaFunctQuintic} with the initial data $\Gamma_0$ satisfying
\begin{equation}
\|\Gamma(t)\|_{C(I,\mathcal{H}^{s})}\leq 2 \|\Gamma_0\|_{\mathcal{H}^{s}}
\end{equation}

\item Given $I_0 = [-T_0, T_0 ]$ with $T_0>0,$ if $\Gamma(t),\Gamma'(t)$  in $C(I_0, \mathcal{H}^{s})$ are two solutions to the Gross-Pitaevskii
hierarchy \eqref{eq:GPHierarchyEquaFunctQuintic} with the initial conditions $\Gamma (0) =\Gamma_0$ and  $\Gamma' ( 0 ) = \Gamma_0'$ in $\mathcal{H}^{s},$ respectively, then
\begin{equation}
\|\Gamma(t)-\Gamma'(t)\|_{C(I, \mathcal{H}^{s})} \leq 2 \|\Gamma_0-\Gamma_0'\|_{\mathcal{H}^{s}}
\end{equation}
with $I = [-T, T],$ where
\be
T = \min \bigg \{ T_0,\; \frac{K_{n,s}}{\|\Gamma(t)-\Gamma'(t)\|_{C(I_0, \mathcal{H}^{s})}^2} \bigg \}.
\ee
\end{enumerate}
\end{theorem}

The proof can be obtained as done in that of Theorem \ref{th:LocalWellposued-alpha>n/2}, based on the following inequality
\be
\| Q^{(k)} \gamma^{(k+2)} \|_{\mathrm{H}^{s}_k} \le C_{n, s} k \| \gamma^{(k+2)} \|_{\mathrm{H}^{s}_{k+2}},\quad \forall k \ge 1,
\ee
with $C_{n, s}>0$ being a constant depending only on $n$ and $s,$ which was proved in \cite[Theorem 4.3]{CP2011}.

For the case $s \le n/2$ we have

\begin{theorem}\label{th:LocalWellposuedQuintic2}
Assume that $n \ge 2$ and $s > \frac{n-1}{2}.$ Then, the Cauchy problem for the Gross-Pitaevskii hierarchy \eqref{eq:GPHierarchyEquaFunctQuintic}  is locally well posed in $\mathcal{H}^{s}.$ More precisely, there exist an absolute constant $A>2$ and a constant $C=M_{n, s}>0$ depending only on $n$ and $s$ such that
\begin{enumerate}[{\rm (1)}]

\item For every $\Gamma_0 = ( \gamma^{(k)}_{0} )_{k \geq 1} \in \mathcal{H}^{s},$ let $I = [- T, T]$ with $T = \frac{M_{n, s}}{\| \Gamma_0 \|^4_{\mathcal{H}^{s}}}.$ Then there exists a solution $\Gamma (t) = ( \gamma^{(k)} (t) )_{k \geq 1} \in C ( I, \mathcal{H}^{s})$ to \eqref{eq:GPHierarchyEquaFunctQuintic} with the initial data $\Gamma (0) = \Gamma_0$ such that
\begin{equation}\label{eq:SpacetimeEstimateQuintic}
\| \hat{Q} \Gamma (t) \|_{L^1_{t \in I} \mathcal{H}^{s}} \leq 2 A \| \Gamma_0 \|_{\mathcal{H}^{s}}.
\end{equation}

\item Given $I_0 = [- T_0, T_0]$ with $T_0 >0,$ if $\Gamma (t) \in C (I_0, \mathcal{H}^{s})$ so that $\hat{Q} \Gamma (t) \in L^1_{t \in I_0} \mathcal{H}^{s}$ is a solution to \eqref{eq:GPHierarchyEquaFunctQuintic} with the initial data $\Gamma (0) = \Gamma_0,$ then \eqref{eq:SpacetimeEstimateQuintic} holds true as well for $I = [-T, T],$ where
\be
T = \min \bigg \{ T_0, \; \frac{M_{n, s}}{ \| \hat{Q} \Gamma (t) \|^4_{L^1_{t \in I_0} \mathcal{H}^{s}} + \| \Gamma_0 \|^4_{\mathcal{H}^{s}} } \bigg \}.
\ee

\item Given $I_0 = [-T_0, T_0]$ with $T_0 >0,$ if $\Gamma (t)$ and $\Gamma' (t)$ in $C ( I_0, \mathcal{H}^{s})$ with $\hat{Q} \Gamma (t), \hat{Q} \Gamma' (t) \in L^1_{t \in I_0}\mathcal{H}^{s}$ are two solutions to \eqref{eq:GPHierarchyEquaFunctQuintic} with initial conditions $\Gamma (0) = \Gamma_0$ and $\Gamma' (0) = \Gamma'_0$ in $\mathcal{H}^{s},$ respectively, then
\begin{equation}\label{eq:InitialValueContinuousDependenceQuintic}
\| \Gamma (t) - \Gamma' (t) \|_{C( I, \mathcal{H}^{s})} \leq (1+ 2 A) \| \Gamma_0 - \Gamma'_0\|_{\mathcal{H}^{s}},
\end{equation}
with $I = [-T, T],$ where
\be
T = \min \bigg \{ T_0,\; \frac{M_{n, s}}{ \| \hat{Q} [\Gamma (t) - \Gamma' (t) ]\|^4_{L^1_{t \in I_0} \mathcal{H}^{s}} + \| \Gamma_0 - \Gamma'_0 \|^4_{\mathcal{H}^{s}} } \bigg \}.
\ee
\end{enumerate}
\end{theorem}

To prove this theorem, we need two preliminary results as follows.

\begin{lemma}\label{le:QOperatorEstimateQuintic}
Assume that $n\geq 2$ and $s> \frac{n-1}{2}.$ Then there exists a constant $C_{n,s}>0$ depending only on $n$ and $s$ such that, for any symmetric $\gamma^{(k+2)}\in {\mathcal S}({\mathbb R}^{(k+2)n}\times{\mathbb R}^{(k+2)n} ),$
\begin{equation}
\|Q^{(k)}_{j} e^{\mathrm{i} t\Delta^{(k+2)}}\gamma^{(k+2)}\|_{L_t^2({\mathrm H}_k^{s})}\leq C_{n,s}\|\gamma^{(k+2)}\|_{{\mathrm H}_{k+2}^{s}}
\end{equation}
for all $k\geq 1,$ where $j=1,2,\ldots,k.$

Consequently, $Q^{(k)}$ can be extended to the space ${\mathrm H}_{k+2}^{s}$ such that
\begin{equation}
\|Q^{(k)}e^{\mathrm{i} t\Delta^{(k+1)}}\gamma^{(k+2)}\|_{L_t^2({\mathrm H}_k^{s})}\leq C_{n,s}k\|\gamma^{(k+2)}\|_{{\mathrm H}_{k+2}^{s}}
\end{equation}
for all $\gamma^{(k+2)}\in {\mathrm H}_{k+2}^{s}.$
\end{lemma}

This lemma was proved first for $n=3$ in \cite{KM2008}, and then in \cite[Proposition A.1]{CP2010} for general case.

For any $\Gamma = (\gamma^{(k)}_t )_{k \ge 1}$ we define
\be\begin{split}
P_{k+2, j} (\Gamma) (t) : = & \int^t_0 d t_2 \int^{t_2}_0 d t_4 \cdots \int^{t_{2(j-1)}}_0 d t_{2 j} e^{\mathrm{i} (t - t_2) \triangle^{(k+2)}} \tilde{Q}^{(k+2)} \cdots \\
& \times e^{\mathrm{i} (t_{2(j-1)} - t_{2 j}) \triangle^{(k+ 2 j)}} \tilde{Q}^{(k+ 2 j)} e^{\mathrm{i} t_{2 j} \triangle^{(k+2(j+1))}} \gamma^{(k+2(j+1))}_{t_{2 j}}
\end{split}\ee
with the convention $t = t_0.$

The following lemma is crucial for the proof of Theorem \ref{th:LocalWellposuedQuintic2}.

\begin{lemma}\label{le:DuhamelEstimateQuintic}
Assume that $n\geq 2$ and $s> \frac{n-1}{2}.$  Then there exists an absolute constant $A>2$ and a constant $C_{n,s}$ depending only on $n$ and $s$ so that the estimates below hold
\begin{enumerate}[{\rm (1)}]

\item  For any $\Gamma_0 = ( \gamma_0^{(k)} )_{k\geq 1}\in \bigotimes_{k=1}^{\infty} {\mathrm H}_{k}^{s},$
\beq\label{eq:QoperatorEstimation1}
\|\tilde{Q}^{(k)} P_{k+2, j} ( e^{\mathrm{i} t \Delta} \Gamma_0) (t)\|_{L_{t\in I}^1{\mathrm H}_{k}^{s}}\leq kA^{k+j}(C_{n,s}T)^{\frac{j+1}{2}}\|\gamma^{(k+2j+2)}_0\|_{{\mathrm H}_{k+2j+2}^{s}}
\eeq
for $k,j\geq 1$ and $T>0,$ where $I =[-T, T].$

\item  For any $T>0$ and $\Gamma(t)= ( \gamma^{(k)}_t )_{k \geq 1}$  with $\gamma^{(k)}_t \in L_{t \in [-T,T]}^1{\mathrm H}_{k}^{s},$
\beq\label{eq:QoperatorEstimation2}
\begin{split}
\|\tilde{Q}^{(k)}& P_{k+2, m} (\Gamma )(t)\|_{L_{t \in I}^1{\mathrm H}_{k}^{s}}\\ &\leq kA^{k+m}(C_{n,s}T)^{\frac{m}{2}}\|Q^{(k+2m)}\gamma^{(k+2m+2)}(t)\|_{L_{t \in I}^1{\mathrm H}_{k+2m}^{s}},
\end{split}
\eeq
for $k,m\geq 1,$ where $I =[-T, T].$

\end{enumerate}
\end{lemma}

Now, by slightly repeating the proof of Theorem \ref{th:LocalWellposued-alpha>(n-1)/2} we can prove Theorem \ref{th:LocalWellposuedQuintic2}. The details are omitted.

\subsection*{Acknowledgement} This research was supported in part by NSFC under Grant No. 11171338, and by National Fundamental Research Program of China under Grant No. 2012CB922102.

\end{document}